\documentclass[11pt]{article}

\usepackage{cite}
\usepackage{amsmath,amssymb,amsfonts,amsthm}
\usepackage{algorithm}
\usepackage[noend]{algpseudocode}
\usepackage{graphicx}
\usepackage{textcomp}
\usepackage{bbm}
\usepackage{bm}
\usepackage{multirow}
\usepackage{hhline}
\usepackage{caption}

\newtheorem{theorem}{Theorem}
\newtheorem{corollary}[theorem]{Corollary}
\newtheorem{lemma}[theorem]{Lemma}
\newtheorem{remark}[theorem]{Remark}

\newtheorem{definition}[theorem]{Definition}
\newtheorem{assumption}[theorem]{Assumption}

\newcommand{\dd}{\mathrm{d}}

\def\BibTeX{{\rm B\kern-.05em{\sc i\kern-.025em b}\kern-.08em
    T\kern-.1667em\lower.7ex\hbox{E}\kern-.125emX}}
\begin{document}

\title{Global Optimization of Relay Placement for 
\\
Seafloor Optical Wireless Networks\thanks{This work was supported in part by
JSPS KAKENHI Grant Number 18K18007.}%
}

\author{
Yoshiaki Inoue\thanks{Y.\ Inoue is with
Department of Information and Communications Technology, Graduate
School of Engineering, Osaka University, Suita 565-0821, Japan
(e-mail: yoshiaki@comm.eng.osaka-u.ac.jp).},
Takahiro Kodama\thanks{T.\ Kodama is with Faculty of Engineering and
Design, Kagawa University,  Takamatsu 761-0396, Japan 
(e-mail: tkodama@eng.kagawa-u.ac.jp).
},
and
Tomotaka Kimura\thanks{T. Kimura is with 
Faculty of Science and Engineering, Doshisha University, Kyotanabe 610-0394, Japan
(e-mail: tomkimur@mail.doshisha.ac.jp)
}
}

\maketitle

\allowdisplaybreaks

\begin{abstract}

Optical wireless communication is a promising technology for
underwater broadband access networks, which are particularly important
for high-resolution environmental monitoring applications.
This paper focuses on a deep-sea monitoring system, where an underwater
optical wireless network is deployed on the seafloor. 
We model such an optical wireless network as a general queueing network and
formulate an optimal relay placement problem, whose objective is to
maximize the stability region of the whole system, i.e., the supremum
of the traffic volume that the network is capable of accommodating.
The formulated optimization problem is further shown to be non-convex, 
so that its global optimization is non-trivial.
In this paper, we develop a global optimization method for this
problem and we provide an efficient algorithm to compute an optimal
solution. Through numerical evaluations, we show that a significant
performance gain can be obtained by using the derived optimal solution. 

\end{abstract}


\section{Introduction}

Real-time monitoring of underwater environments, such
as ocean trenches and submarine volcanoes, is of great importance for
scientific research toward the prevention and mitigation of natural disasters.
In such monitoring applications, underwater wireless communication is
a key enabling technology for bringing data from seafloor sensors to
terrestrial base stations \cite{Fele15}. 
Traditionally, acoustic signals have been the primary medium for
underwater wireless communications due to their ability to propagate
over long distances with little energy dissipation.
However, the main weakness of the acoustic channel is the quite
limited data transmission capacity, which is inherent in the use of
kHz-class carrier frequencies. Therefore, acoustic-based underwater
communication networks cannot accommodate the large amount of traffic
generated by high-specification sensors such as underwater LIDARs and
video cameras \cite{Akyi05}, which will be essential for near-future
real-time underwater monitoring systems.

Underwater optical wireless communication (UOWC) is a promising
solution to this problem, which can achieve data rates of several
hundred Mbps to about ten Gbps, provided that the transmission range
is limited to tens to hundreds of meters \cite{Kaushal,Zeng17}.
Because of this limitation on the propagation distance, it is
necessary for the practical use of UOWC to construct 
\textit{a networked optical wireless infrastructure} consisting of
multiple relay nodes. Such an underwater network is called an
underwater optical wireless network (UOWN), and its optimal design has
become a major challenge for realizing underwater real-time monitoring
applications.
Although a wired link (optical fiber) can also be considered as a
connection method between relay nodes, this paper focuses on a relay
system that is interconnected with wireless optical communication,
because the ease of relocation provides operational flexibility
desirable for seafloor monitoring systems that are currently under
development.

\subsubsection*{Motivation} 
Most previous works on UOWNs assume \textit{vertical} network
architectures \cite{Zeng17,Celik20,Celik19,Saeed19-2,Saeed19-3,Saeed19-4},
where data packets generated by seafloor sensors are
transferred to a terrestrial base station in multi-hop fashion via
vertically deployed optical wireless relay nodes.
In such a vertical network architecture, autonomous underwater
vehicles (AUVs) hovering in the water are inevitably used as relay
nodes in addition to those anchored to the seafloor. 

Such architectures with relay AUVs are targeted at
relatively shallow marine environments with depths not exceeding 1000
meters, and their use for \textit{deep-sea monitoring} is impractical
due to the following two reasons. 
Firstly, the monitoring of deep-sea environments with a vertical
network requires a very large number of AUV relay nodes to connect
seafloor sensors to nodes at the sea surface, resulting in enormous
costs. Secondly, the AUV relay nodes must be controlled to keep
hovering in the turbulent water, making it difficult to keep all the
links stable. To the best of our knowledge, there has not been
sufficient attention paid to investigating network architectures that
can solve these problems of deep-sea environment monitoring.

\subsubsection*{Contributions}
This paper proposes \textit{a seafloor optical wireless network (SOWN)}, 
which enables efficient data acquisition from deep-seafloor
environments without employing hovering AUV relay nodes. The main
components of the proposed SOWN are (i) a
terrestrial base station, (ii) a sink node on the seafloor connected
to the terrestrial base station with an optical fiber, and (iii)
anchored relay nodes horizontally deployed on the seafloor;
see Fig.~\ref{fig:network_image} for an illustration. 
The SOWN serves as an infrastructure to accommodate data traffic
originating from a variety of sensors on the seafloor. The sensing data
generated by each sensor is first collected at
the nearest relay node, then delivered to the sink node by optical
wireless multi-hop transmission, and then transferred to the
terrestrial base station via the optical fiber.
Its main advantage being constructed without hovering AUV nodes,
the SOWN is a suitable network architecture for deep-sea monitoring
systems in terms of the cost-effectiveness and 
stability.

\begin{figure}
 \centering
 \includegraphics[width=0.88\linewidth]{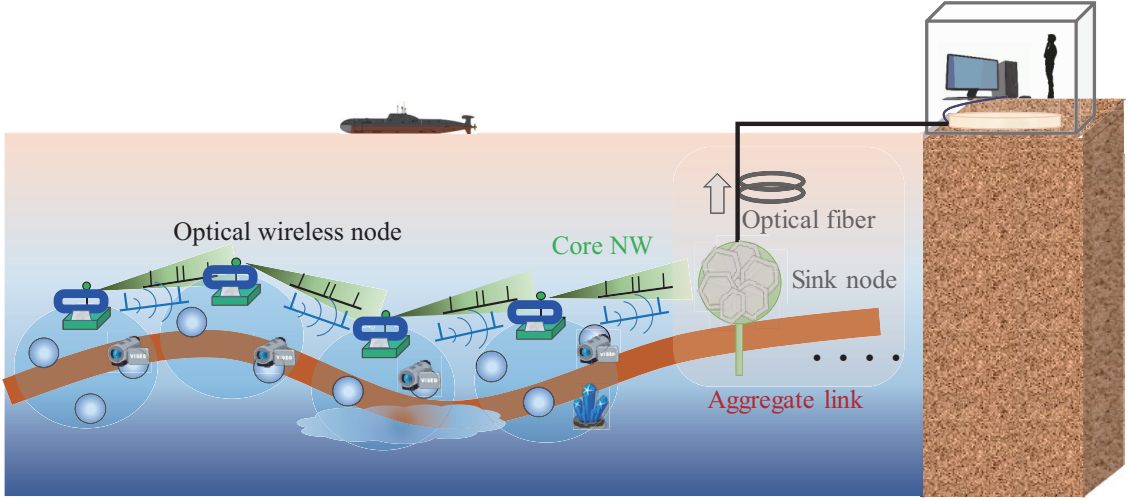}
 \caption{A seafloor optical wireless network.}
 \label{fig:network_image}
\end{figure}

The main focus of this paper is on the development of \textit{an
optimal relay placement method}, which is the most fundamental
challenge toward the optimal design of the SOWN.
Underwater relay-node placement problems have
traditionally been discussed for acoustic-based networks
\cite{Kam14,Sou16,Liu17,Pras18}, where it is
known to be optimal to use a constant relay spacing \cite{Kam14},
provided that the carrier frequency is appropriately selected.
The key observation in this paper, however, is that such a
constant-spaced relay placement cannot fully extract the
transmission capacity of the whole network in the SOWN, but rather
a placement with \textit{optimally determined non-constant node spacing}
significantly improves the network performance.
Such an improvement basically stems from the fact that the capacity of
an underwater optical channel is significantly affected by the node
distance, due to the rapid attenuation of the optical signal with
propagation distance \cite{Giles05,Doniec13,Elam19}. In order to efficiently
utilize the resources of the whole network, it is then necessary to
arrange relay nodes in such a way that \textit{the distance between
the nodes gradually increases} from the sink node to the end (leaf) node,
because optical wireless links close to the sink node have
to relay a large amount of sensing data transferred from upstream
nodes and require a larger channel capacity than those away from the
sink node. 

In this paper, we make this idea concrete by modeling the relay
placement in the SOWN mathematically and performing its detailed analysis.
More specifically, we first introduce a queueing-network model whose
input process differs depending on the relay-node placement, under a
mild assumption that the packet generation follows a general
stationary point process. 
Using this model, we then formulate an optimal relay placement problem
that aims to maximize \textit{the stability region} of the whole system.
The stability region is defined as the range of total traffic load
that the network can accommodate without exceeding the capacity of any
communication links, which is of primary importance in designing
communication networks because it determines the fundamental
performance limit of the system as well known in the queueing theory.

The main technical challenge we have to address in this relay placement
problem is that the formulated optimization problem is inherently
\textit{non-convex}, as will be shown later.
Therefore its global optimization is non-trivial, and general-purpose
off-the-shelf algorithms can basically yield only local optimal values.
In this paper, we perform a detailed theoretical analysis of the
optimal relay placement problem and develop a global optimization
algorithm that can be executed quite efficiently.

As an initial study of the relay placement problem in the SOWN, 
this paper mainly focuses on a one-dimensional network, i.e., the 
case where relay nodes are placed along a straight line.
Although this assumption may restrict the direct applicability of the
results to be obtained, this simplification allows us to reveal the
exact structural properties of the global optimal solution, as will be
shown in this paper. Since the one-dimensional network is a
fundamental building block of a more general two- or three-dimensional
UOWNs, the mathematical analysis developed in this paper also provides
theoretical insights into the network design of such general UOWNs;
we shall later demonstrate how the mathematical results obtained for
the one-dimensional network can be extended to the two-dimensional
case. It is also worth noting that the one-dimensional SOWN itself has
important practical applications for mitigating natural disasters 
(particularly earthquakes), such as high-resolution real-time
monitoring of ocean trenches.

\subsubsection*{Organization}
The rest of this paper is organized as follows. In Section
\ref{sec:related_works}, we provide a brief review of previous studies
related to UOWNs and relay placement problems. In Section
\ref{sec:model}, we introduce a queueing-network model representing
the SOWN and formulate an optimal relay placement problem based on it.
In Section \ref{sec:global_optimization}, we develop a global
optimization method for the relay placement problem and investigate
the mathematical structure of the obtained optimal solution. In Section
\ref{sec:numerical}, we first examine the performance of the obtained
optimal solution through extensive numerical experiments. 
In particular, we show that the optimal placement with non-constant
spacing significantly improves the system performance compared to the
constant spacing case. We then demonstrate an extension of the
obtained results to a two-dimensional seafloor network.
Finally, this paper is concluded in Section \ref{sec:conclusions}.

\section{Related works}
\label{sec:related_works}

Gbps-class transmission capacity in UOWC has been achieved by using
visible light bands, where the effects of absorption and scattering
losses are relatively small.
UOWC is still in the early stages of development, and several
demonstration experiments have been carried out in recent years
\cite{Nakamura,Oubei,A-Halafi,Wu17,Kodama}.
On the other hand, theoretical investigations on UOWC channel
characteristics have been carried out from earlier years, and 
various channel models have been proposed.
Giles and Bankman \cite{Giles05} derived a basic signal-to-noise ratio
(SNR) formula for UOWC channels, which was further extended to an
end-to-end signal strength model by Doniec et al.\ \cite{Doniec13},
where its validity was confirmed in a real system.
Elamassie et al.\ \cite{Elam19} have also extended this SNR formula
and proposed a correction that takes into account the contribution of
the scattered light that partially reaches the detector.
For a more detailed characterization of the UOWC channel, 
Tang et al.\ \cite{Tang} have proposed a channel impulse response
model with a double gamma function. Jaruwatanadilok \cite{Jaruwatanadilok}
has developed a channel model based on radiative transfer theory as
well and Zhang et al.\ \cite{Zhang} have presented a stochastic
channel model representing the spatiotemporal probability distribution
of propagating photons, taking into account the non-scattering,
single-scattering, and multiple-scattering components.

From the perspective of UOWC networking, Akhoundi et al.\ \cite{Akhoundi15} 
have introduced an optical code-division multiple access (CDMA) underwater
cellular network and evaluated its performance in several water types.
Optical CDMA underwater networks have been further studied by Jamali
et al.\ \cite{Jamali16}, reflecting the turbulent behavior of
underwater channels. Jamali et al.\ \cite{Jamali17} have also
presented the benefits of serial relayed multi-hop transmission using
a bit detection and transfer (BDF) strategy, showing that multi-hop
transmission can significantly improve system performance by
mitigating adverse effects on all channels.
Vavoulas et al.\ \cite{Vavoulas14} have studied an effective path loss
model in UOWC and characterized the connectivity of long-distance
underwater communications.
Saeed et al.\ have analyzed network localization performances 
in terms of the network connectivity in \cite{Saeed19-3} 
and proposed a localization framework for energy
harvesting nodes in \cite{Saeed19-4}. In \cite{Saeed19-1}, they have
also discussed an optimal placement of seaface anchor nodes in terms
of the localization accuracy.
To evaluate the performance of a video streaming under the sea,
Al-Halafi et al.\ \cite{Al-Ha19} have modeled UOWC channels with M/G/1
queues, assuming that there are multiple laser diodes in
the transmitter and multiple avalanche photodiodes in the receiver.
Celik et al.\ \cite{Celik20} have analyzed the end-to-end bit error
rates for the decode and forward (DF) and amplify and forward (AF)
relaying in a vertical UOWN. Furthermore, in \cite{Celik19}, a
sector-based opportunistic routing protocol has been devised where
packets are transmitted simultaneously to multiple relay nodes that
fall within the range of a directional beam.
Xing et al.\ \cite{Xing20} have investigated problems of minimizing
energy consumption and maximizing SNR by performing relay node
selection and power allocation simultaneously in the AF scheme.

As mentioned earlier, relay-node placement under water has been
studied in the context of acoustic communication systems.
Kam et al.\ \cite{Kam14} considered a problem of optimizing the
frequency and node location to minimize the energy consumption.
Souza et al.\ \cite{Sou16} considered the minimization of energy
consumption taking into account the optimal number of hops,
retransmission, coding rate, and SNR, where the distance between nodes
of each hop is assumed to be constant.
Liu et al.\ \cite{Liu17} have developed flow assignment and relay node
placement methods in a vertical UOWN to maximize network lifetime,
where it is assumed that relay nodes are fixed in horizontal
coordinates and can be changed only in vertical coordinates.
Prasad et al.\ \cite{Pras18} have discussed a problem for a two-hop
network that minimizes the probability of receiving power falling
below an outage-data-rate threshold by properly controlling the locations of
relay nodes and the transmission power.

\section{Model}
\label{sec:model}

Throughout the paper, we follow the convention that for any
$k$-dimensional ($k=1,2,\ldots$) vector $\bm{y} \in \mathbb{R}^k$, its
$i$th element is denoted by
$y_i$. We further define empty sum terms as zero.

Let $\mathcal{N} = \{1,2,\ldots,N\}$ ($N = 1,2,\ldots$) denote the set
of relay nodes. Relay nodes are aligned on a subset $\mathcal{L} := [0, L]$
of the real half-line $\mathbb{R}^+$, and the sink node is placed at
the origin $x=0$. Let $x_n$ ($n = 1,2,\ldots,N$) denote
the position of the $n$th node. We assume $0 \leq x_n \leq
x_{n+1}$ ($n = 1,2,\ldots,N-1$) without loss of generality. We call
the sink node `the $0$th' node, so that $x_0 := 0$ is defined
accordingly. We assume that $x_N = L$ holds and that the
one-dimensional region $\mathcal{L}$ is completely covered by the sink
node and $N$ relay nodes as described below.

We assume that generation times of data packets follow a general
stationary point process and that the generation points of those
packets are uniformly distributed on $\mathcal{L}$. Each packet is
first collected by the nearest node from its generation point and
then transferred to the sink node with multi-hop transmissions. 
More formally, we define the coverage area $\mathcal{C}_n \subseteq \mathbb{R}^+$ ($n
= 0,1,\ldots,N$) of the $n$th node as its Voronoi cell, 
which is given by a half-open interval $\mathcal{C}_n = [a_n, b_n)$
with
\begin{equation}
a_0 = 0,
\;\;
a_n = \frac{x_{n-1}+x_n}{2},
\;
n = 1,2,\ldots,N,
\quad
b_n = a_{n+1}, 
\;
n = 0,1,\ldots, N-1,
\;\;
b_N = x_N.
\label{eq:C_n}
\end{equation}
See Fig.\ \ref{fig:covering} for an illustration.
Clearly we have $\cup_{n=0}^N \mathcal{C}_n = [0, x_N)$ and 
$\mathcal{C}_i \cap  \mathcal{C}_j = \emptyset$ for $i \neq j$.
We further assume that packet transmissions are performed in the
store-and-forward manner (DF relaying, in other words). 
The system is then represented as a network of $N$ G/G/1 queues
depicted in Fig.\ \ref{fig:GG1}, where $\lambda$ denotes the
mean number of generated packets per unit time (within the whole covered
area $\mathcal{L}$) and $B$ denotes the mean data size.

\begin{figure}[tbp]
\centering
\begin{minipage}[t]{0.48\textwidth}
\includegraphics[scale=0.5]{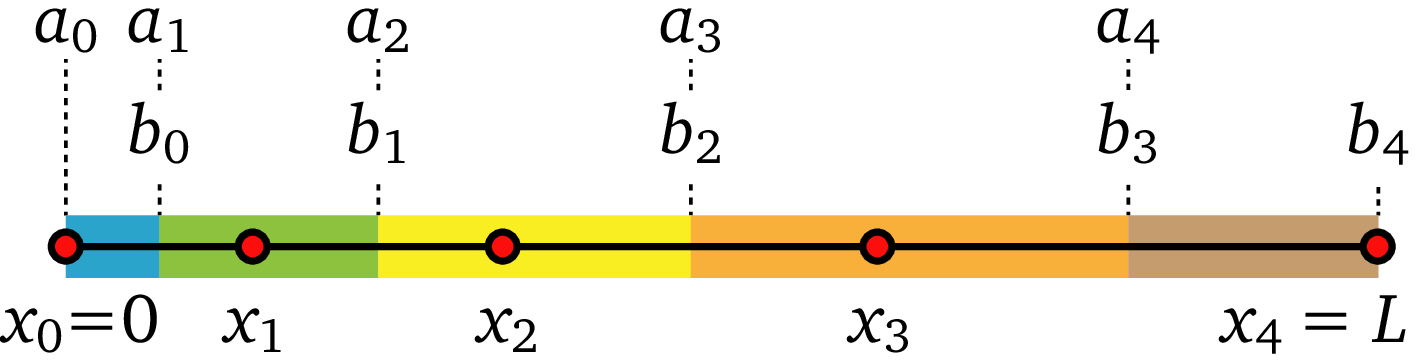}
\vspace{-2ex}
\caption{An illustration of the system model ($N=4$).}
\label{fig:covering}
\end{minipage}
\;\;
\begin{minipage}[t]{0.48\textwidth}
\includegraphics[scale=0.5]{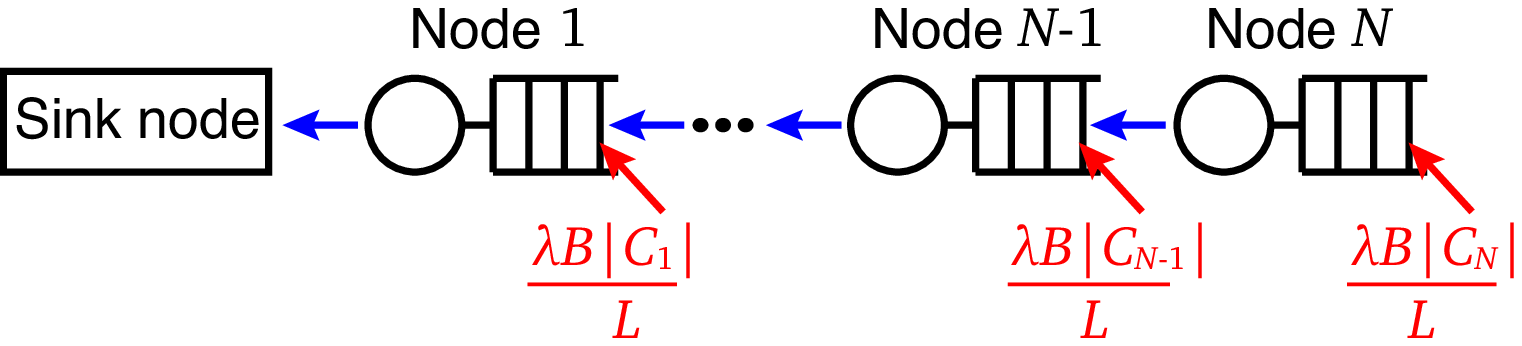}
\vspace{-2ex}
\caption{The SOWN modeled as a network of G/G/1 queues.}
\label{fig:GG1}
\end{minipage}
\end{figure}

We define $\rho_n$ ($n = 1,2,\ldots,N$) as the traffic intensity of
external arrivals to the $n$th node, i.e.,
\begin{equation}
\rho_n = \frac{\lambda |\mathcal{C}_n|}{L} \cdot B
= q |\mathcal{C}_n|,
\qquad
q := \frac{\lambda B}{L}.
\label{eq:rho_n-def}
\end{equation}
Observe that $q$ represents the amount of data brought into the
system per unit time, normalized by the area length.
Owing to \cite[Page 142]{Baccelli03}, the stability condition of 
this system is given by that for each node $i$, the total traffic
intensity of relayed packets does not exceed the link capacity:
\begin{equation}
\sum_{n=i}^N \rho_n < R(d_i),
\;\;
i \in \mathcal{N},
\label{eq:stability}
\end{equation}
where $d_n := x_n - x_{n-1}$ ($n = 1,2,\ldots$) denotes the
distance between the ($n-1$)st and $n$th nodes, and
$R(d)$ ($d \geq 0$) denotes \textit{the effective transmission
rate} between two nodes with distance $d$, which is formulated as follows. 

Let $\mathrm{SNR}(d)$ ($d \geq 0$) denote the electrical SNR at distance $d$. 
A widely used model \cite{Giles05,Doniec13,Elam19,Xing20,Arnon09}
for representing the SNR of a UOWC channel is that
the $\mathrm{SNR}(d)$ takes a form proportional to
$d^{-\alpha}e^{-Kd}$ for some coefficients $\alpha > 0$ and $K > 0$.
In this expression, $d^{-\alpha}$ represents the signal attenuation
due to the geometric spreading of the light beam and $\alpha=2$ is
usually used to represent the spherical spreading.
On the other hand, $e^{-Kd}$ represents the contribution of absorption
and scattering losses, and $K$ is given by the sum of the absorption
and scattering coefficients, which vary depending on the type of water
and the light wavelength. Readers are referred to \cite{Kaushal,Saeed19-2,Doniec13,Mobley}
for more detailed explanations on such a theoretical characterization
and its validation in a real system.

In this paper, to avoid the singularity of $d^{-\alpha}$ at the
origin, we consider the following bounded expression for
$\mathrm{SNR}(d)$, with a small $\epsilon > 0$:
\begin{equation}
\mathrm{SNR}(d)
= 
Ae^{-Kd}(\epsilon+d)^{-\alpha},
\label{eq:SNR}
\end{equation}
where $A$ denotes a constant that depends on physical parameters (an
example will be given later in Section \ref{sec:numerical}).
It should be noted here that $\epsilon$ does not have a specific
physical meaning: it is a parameter intended to correct the singular
behavior that $d^{-\alpha}$ diverges near the origin, and the value of
$\epsilon$ has little effect on $\mathrm{SNR}(d)$ unless $d$ is very small
(such a correction term is often used in the radio communication
literature \cite{Andrews11}).
Owing to the Shannon-Hartley theorem, with $W$ denoting the bandwidth,
the effective transmission rate $R(d)$ ($d \geq 0$) is then expressed as
\begin{equation}
R(d) = W\log(1 + \mathrm{SNR}(d)).
\label{eq:R-Shannon}
\end{equation}

In order not to restrict the applicability of our theoretical results,
however, we do not assume any specific expression for the function
$R(d)$ in performing mathematical analysis below.
Instead, we make only the following assumption on $R(d)$, which is
clearly satisfied by (\ref{eq:SNR}) and (\ref{eq:R-Shannon}):
\begin{assumption} \label{assumption:R-convex}
The effective transmission rate $R: [0,\infty) \rightarrow
[0,\infty)$ is a strictly decreasing, continuously differentiable
convex function of the node distance, and $\lim_{d \to \infty}R(d) = 0$.
\end{assumption}

\begin{remark}
Another example of an expression for $R(d)$ (other than the Shannon capacity
(\ref{eq:R-Shannon})) is given as follows.
Suppose that sensing information is coded and modulated with
(i) the modulation level $M$ [bits/symbol] and
(ii) a forward-error-correction (FEC) code with code rate $\eta$ ($0 <
\eta < 1$). 
Also suppose that the FEC code enables the receiver to decode the
signal with a negligible error-rate, provided that the SNR does not
exceed a threshold $\zeta$. This is an abstraction of UOWC channels
implemented with standard modulation techniques, such as the on-off keying (OOK)  
and the quadrature amplitude modulation (QAM) \cite{Lavery19}.

In this setting, it is reasonable that the transmitter uses the
maximum symbol rate (which equals the bandwidth $W$) such that the
constraint
$\mathrm{SNR}(d) \leq \zeta$ for error-free transmission is satisfied.
Assuming that the noise spectral density is constant 
(i.e., white noise) over the operating frequency range, the
expression (\ref{eq:SNR}) is rewritten as
$
\mathrm{SNR}(d) = A'e^{-Kd}(\epsilon + d)^{-\alpha}W^{-1},
$
where $A'$ does not depend on the symbol rate $W$. 
$\mathrm{SNR}(d)$ then decreases with $W$, so that the maximum symbol
rate is achieved if (and only if) $\mathrm{SNR}(d) = \zeta$, i.e.,
$
W = A'e^{-Kd}(\epsilon + d)^{-\alpha}\zeta^{-1}.
$
As the effective transmission rate equals $\eta \cdot M \cdot W$,
we then conclude
$
R(d) = \eta M A'e^{-Kd}(\epsilon + d)^{-\alpha}\zeta^{-1},
$
which clearly satisfies Assumption \ref{assumption:R-convex}.
\end{remark}

\begin{remark}
A refinement of the SNR equation correcting the exponential term
as $e^{-Kd^{\beta}}$ ($\beta \in (0,1]$) is proposed in
\cite{Elam19}; the discussion above is still valid under such an extension.
\end{remark}

We see from (\ref{eq:rho_n-def}) and (\ref{eq:stability}) that the
stability region of the system varies depending on the node placement 
$\bm{x} := (x_1, x_2,\ldots, x_N)^{\top}$. Let $\rho_n(q,\bm{x})$
($n=1,2,\ldots,N$) denote the traffic intensity $\rho_n$ of external
arrivals to the $n$th node, represented as a function of the
normalized traffic intensity $q$ and the placement of relay nodes $\bm{x}$ 
(cf.\ (\ref{eq:rho_n-def})). The size of the stability region is
characterized by \textit{the normalized throughput limit}
$q_{\sup}(\bm{x})$, which is defined as the least upper bound of the
normalized throughput $q$ for which the system is stable:
\begin{align}
q_{\sup}(\bm{x}) 
&= 
\sup\biggl\{q \in \mathbb{R}^+ \ \biggl|\
\sum_{n=i}^N \rho_n(q,\bm{x}) < R(x_i-x_{i-1}),\,
i \in \mathcal{N}
\biggr\}
\nonumber
\\
&=
\max\biggl\{q \in \mathbb{R}^+ \ \biggl|\
\sum_{n=i}^N \rho_n(q,\bm{x}) \leq R(x_i-x_{i-1}),\,
i \in \mathcal{N}
\biggr\}.
\label{eq:q_sup-def}
\end{align}

The size of the stability region (i.e., the value of the normalized
throughput limit $q_{\sup}(\bm{x})$) is the most fundamental
performance metric in designing the communication network. 
In this paper, we thus employ $q_{\sup}(\bm{x})$ as the objective
function of our optimal placement problem. Specifically, we
develop a solution method to the following optimization problem:
\begin{align*}
\underset{\bm{x} \in \mathbb{R}^N}{\mathrm{maximize}}
\;\;
q_{\sup}(\bm{x})
\quad
\mathrm{s.t.}
\;\;
x_N = L,
\;\;
x_{i+1} \geq x_i \geq 0,
\;\;
i \in \{1,2,\ldots,N-1\},
\end{align*}
which is rewritten as (cf.\ (\ref{eq:q_sup-def}))
\begin{align}
\underset{q \in \mathbb{R},\, \bm{x} \in \mathbb{R}^N}{\mathrm{maximize}}
\;\;
q
\quad
\mathrm{s.t.}\;\;&
\sum_{n=i}^N \rho_n(q,\bm{x}) \leq R(x_i-x_{i-1}),
\;\;
i \in \{1,2,\ldots,N\},
\nonumber
\\
&
x_N = L,
\;\;
x_{i+1} \geq x_i,
\;\;
i \in \{1,2,\ldots,N-1\},
\nonumber
\\
&
q \geq 0,
\;\;
x_i \geq 0,
\;\;
i \in \{1,2,\ldots,N\}.
\tag{U${}_0$}
\label{opt:bandwidth}
\end{align}
Note that an optimal solution of (\ref{opt:bandwidth}) provides 
not only an optimal placement $\bm{x}^*$ of relay nodes but also
the achievable maximum value of the normalized throughput limit $q_{\sup}^* := q_{\sup}(\bm{x}^*)$.

\section{Global optimization method}
\label{sec:global_optimization}

In this section, we develop a global optimization method for
(\ref{opt:bandwidth}). We start with rewriting (\ref{opt:bandwidth})
into a more comprehensive form. We have from 
(\ref{eq:C_n}) and (\ref{eq:rho_n-def}),
\begin{equation}
\rho_n(q,\bm{x}) = \frac{q(x_{n+1}-x_{n-1})}{2},
\;\;
n = 1,2,\ldots,N-1,
\quad
\rho_N(q,\bm{x}) = \frac{q(x_N - x_{N-1})}{2},
\end{equation}
so that we can rewrite (\ref{opt:bandwidth}) in terms of the distance
$d_n = x_n-x_{n-1}$ ($n = 1,2,\ldots,N$) of nodes:
\begin{align}
\underset{q \in \mathbb{R},\, \bm{d} \in \mathbb{R}^N}{\mathrm{maximize}}
\;\;
q
\quad
\mathrm{s.t.}
\;\;
&
R(d_i)
-
\frac{qd_i}{2}
-
\sum_{n=i+1}^N qd_n \geq 0,
\;\;
i \in \{1,2,\ldots,N\},
\nonumber
\\
& 
q \geq 0, \;\;
\sum_{i=1}^N d_i = L, 
\;\;
d_i \geq 0,
\;\;
i \in \{1,2,\ldots,N\}.
\tag{U}
\label{opt:original}
\end{align}

It is readily verified that (\ref{opt:original}) does not have a
convex feasible region: for fixed $q \geq 0$, each inequality
constraint takes the form that a convex function is not less than $0$
(cf.\ Assumption \ref{assumption:R-convex}), 
so that its feasible region is the complement of a convex set.
Such constraints are known as \textit{reverse-convex} constraints,
and a class of optimization problems with this type of constraints is 
called the reverse-convex programming (RCP) \cite{Hillestad80}. In
general, an optimization problem
\begin{equation}
\underset{\bm{y} \in \mathbb{R}^K}{\mathrm{maximize}}\; u(\bm{y})
\quad
\mathrm{s.t.}\;
f_i(\bm{y}) \geq 0,
\;\;
i = 1,2,\ldots,M,
\tag{R}
\label{opt:RCP}
\end{equation}
with $K$ variables and $M$ constraints ($M \geq K$) is called RCP if
$u$ and $f_i$ ($i = 1,2,\ldots,M$) are quasi-convex.
Note here that any equality constraints of the form $\sum_{k=1}^K w_k y_k = c$ 
($c \in \mathbb{R}$, $w_k \in \mathbb{R}$) can be
translated into the double number of linear (thus reverse-convex)
inequality constraints $\sum_{k=1}^K w_k y_k \leq c$ and
$\sum_{k=1}^K w_k y_k \geq c$. 
Due to the non-convexity of its feasible region,
global optimization for RCP is not an easy task in general, and various
algorithms to find a globally optimal solution have been developed in the 
literature (see e.g., \cite{Hillestad80,Ueing72,Tuy87,Moshirvaziri11} and
references therein). Here we introduce a known theoretical property of RCP,
which will be used in our analysis. 
Let $\mathcal{A} := \{\bm{y};\, f_i(\bm{y}) \geq 0\,
(i=1,2,\ldots,M)\}$ denote the feasible region of (\ref{opt:RCP})
and let $\mathcal{I}(\bm{y}) := \{i \in \{1,2,\ldots,M\};\,
f_i(\bm{y})=0\}$ ($\bm{y} \in \mathcal{A}$).
\begin{definition}[\!\!{\cite[Def.\ 1]{Hillestad80}}]
$\bar{\bm{y}} \in \mathcal{A}$ is called a basic solution of
(\ref{opt:RCP}) if the matrix with row vectors 
$\{\nabla f_i(\bar{\bm{y}});\, i \in \mathcal{I}(\bar{\bm{y}})\}$ 
has rank $K$.
\end{definition}

\begin{lemma}[\!\!{\cite[Th.\ 9]{Hillestad80}}] \label{lemma:RCP-basic}
If $u$ and $f_i$ ($i = 1,2,\ldots,M$) are quasi-convex,
(\ref{opt:RCP}) has an optimal solution which is also basic.
\end{lemma}

\begin{remark} \label{remark:activate}
A basic solution must satisfy at least $K$ constraints with
equality. Lemma \ref{lemma:RCP-basic} thus implies that
there exists an optimal solution that satisfies
at least $K$ constraints with equality.
\end{remark}

In what follows, we develop a global optimization method
(\ref{opt:original}) utilizing its special structure. To that end, we
first introduce the following subproblem for each $q > 0$:
\begin{align}
\underset{\bm{d} \in \mathbb{R}^N}{\mathrm{maximize}}
\;\;
\sum_{n=1}^N d_n
\quad
\mathrm{s.t.}
\;\;
&
R(d_i)
-
\frac{qd_i}{2}
-
\sum_{n=i+1}^N qd_n \geq 0,
\;\;
i \in \{1,2,\ldots,N\},
\nonumber
\\
& 
d_i \geq 0,
\;\;
i \in \{1,2,\ldots,N\}.
\tag{S$_q$}
\label{opt:subproblem}
\end{align}
The main difference between (\ref{opt:original}) and
(\ref{opt:subproblem}) is that $q$ is not variable but fixed
in (\ref{opt:subproblem}). Also, the coverage $\sum_{i=1}^N d_i$
of relay nodes is to be maximized in (\ref{opt:subproblem}),
while it is fixed to be $L$ in (\ref{opt:original}).

It is readily verified that the subproblem (\ref{opt:subproblem})
still belongs to RCP. Owing to its special structure, however, a
globally optimal solution of (\ref{opt:subproblem}) is explicitly
obtained. For a fixed $q > 0$, we define a function 
$g_q: [0,\infty) \to (-\infty, R(0)/q]$ as
\begin{equation}
g_q(x) = \frac{R(x)}{q} - \frac{x}{2},
\quad
x \geq 0.
\label{eq:g_q-def}
\end{equation}
From Assumption \ref{assumption:R-convex}, $g_q$ is
a strictly decreasing continuous function, so that it has a
unique inverse function $g_q^{-1}: (-\infty, R(0)/q] \to [0,\infty)$. 
The following results show that an optimal solution of 
(\ref{opt:subproblem}) is explicitly constructed in terms of
$g_q$ and $g_q^{-1}$:

\begin{theorem} \label{theorem:optimal_solution_Sq}
\begin{itemize}
\item[(i)] If $g_q^{-1}(0) \geq g_q(0)$, then 
the following $\bm{d}^*$ is an optimal solution of (\ref{opt:subproblem}):
\begin{equation}
\bm{d}^* = (0, 0, \ldots, 0, g_q^{-1}(0))^{\top} \in \mathbb{R}^k.
\label{eq:d^*0}
\end{equation}

\item[(ii)] If $g_q^{-1}(0) < g_q(0)$, a backward recursion
\begin{equation}
d_N^* = g_q^{-1}(0), 
\quad
d_i^* = g_q^{-1}\left(\sum_{n=i+1}^N d_n^*\right),
\;\;
i = 1,2,\ldots,N-1
\label{eq:d-recursion}
\end{equation}
well-defines $d_1^*, d_2^*, \ldots, d_N^*$ such that
\begin{equation}
0 < d_i^* < d_{i+1}^*,
\quad
i = 1,2,\ldots,N-1,
\end{equation}
and the following $\bm{d}^*$ is an optimal solution of (\ref{opt:subproblem}):
\begin{equation}
\bm{d}^* = (d_1^*, d_2^*, \ldots, d_N^*)^{\top} \in \mathbb{R}^N.
\label{eq:d^*}
\end{equation}
\end{itemize}
\end{theorem}

\begin{remark}\label{remark:theorem:optimal_solution_Sq}
The proof of Theorem \ref{theorem:optimal_solution_Sq} is somewhat
complicated because a careful treatment is necessary to differentiate
between the two cases (i) and (ii).  Basically, our proof is based on
the fact mentioned in Remark \ref{remark:activate} that there exists
an optimal solution satisfying at least $N$ constraints with equality. 
Recall that (\ref{opt:subproblem}) has $2N$ constraints: $N$ out of
these are of the form $g_q(d_i) \geq \sum_{n=i+1}^N d_n$ and the
others are of the form $d_i \geq 0$. The solution (\ref{eq:d^*})
satisfies the former with equality for $i=1,2,\ldots,N$ and it has
all non-zero elements. On the other hand, the solution (\ref{eq:d^*0})
has only one non-zero element, and it satisfies $g_q(d_i) > \sum_{n=i+1}^N
d_n$ for $i=2,3,\ldots,N$.
\end{remark}
\noindent
We provide the proof of Theorem \ref{theorem:optimal_solution_Sq}
in Appendix \ref{appendix:optimal_solution_Sq}. 

The optimal solution given in Theorem \ref{theorem:optimal_solution_Sq} 
takes a different form depending on whether or not $g_q^{-1}(0) \geq
g_q(0)$ holds. While it is easy to check if this inequality
holds for given $q$, we also have a simple criterion shown in the
following lemma, which is useful in the theoretical analysis below:
\begin{lemma}
\label{lemma:g_q-changepoint-q}

Let $q_0$ denote the unique solution of 
\begin{equation}
R\left(\frac{R(0)}{q}\right) - \frac{R(0)}{2} = 0,
\quad
q > 0.
\label{eq:R-equation}
\end{equation}
We then have
\begin{equation}
g_q^{-1}(0) \geq g_q(0),
\;\;
\forall q \geq q_0,
\qquad
g_q^{-1}(0) < g_q(0),
\;\;
\forall q < q_0.
\end{equation}
\end{lemma}
\begin{proof}
Because $g_q(x)$ is strictly decreasing with respect to $x$,
\begin{align}
g_q^{-1}(0) \geq g_q(0)
\ \Leftrightarrow\ 
g_q(g_q(0)) \leq 0
\ \Leftrightarrow\
R\left(\frac{R(0)}{q}\right) - \frac{R(0)}{2} \leq 0,
\label{eq:g-condition-R}
\end{align}
where we used (\ref{eq:g_q-def}) and $q > 0$ in the last equality.
From Assumption \ref{assumption:R-convex}, $R(R(0)/q) - R(0)/2$
is continuous and strictly increasing in $q$ and
\begin{equation}
\lim_{q \to 0+} 
\left[
R\left(\frac{R(0)}{q}\right) - \frac{R(0)}{2} 
\right]
= 
- \frac{R(0)}{2} < 0,
\quad
\lim_{q \to \infty} 
\left[
R\left(\frac{R(0)}{q}\right) - \frac{R(0)}{2} 
\right]
= 
\frac{R(0)}{2} > 0,
\end{equation}
so that (\ref{eq:R-equation}) has the unique solution.
Lemma \ref{lemma:g_q-changepoint-q} now follows immediately
from (\ref{eq:g-condition-R}).
\end{proof}

We then relate the subproblem (\ref{opt:subproblem}) with the original
problem (\ref{opt:original}). For $q > 0$, let $\bm{d}_q^* := (d_{q,1}^*,
d_{q,2}^*,\ldots,d_{q,N}^*)^{\top}$ denote the optimal solution of
(\ref{opt:subproblem}) given in Theorem \ref{theorem:optimal_solution_Sq}, 
and let $x_{q,N}^* := \sum_{i=1}^N d_{q,i}^*$
denote the corresponding optimal value. The following theorem shows that we can
obtain a globally optimal solution of (\ref{opt:original}) by
iteratively solving (\ref{opt:subproblem}):

\begin{theorem} \label{theorem:subproblem}
\begin{itemize}
\item[(a)] The optimal value $x_{q,N}^*$ of (\ref{opt:subproblem}) is
a continuous, strictly decreasing function of $q$ with
$\lim_{q \to 0+} x_{q,N}^* = \infty$ and $\lim_{q \to \infty}
x_{q,N}^* = 0$.
\item[(b)] The optimal value $q_{\sup}^*$ of (\ref{opt:original}) is
characterized as follows:
\begin{equation}
x_{q,N}^* > L \ \Leftrightarrow\ q < q_{\sup}^*,
\quad\;
x_{q,N}^* < L \ \Leftrightarrow\ q > q_{\sup}^*,
\quad\;
x_{q,N}^* = L \ \Leftrightarrow\ q = q_{\sup}^*. 
\label{eq:subproblem}
\end{equation}
Furthermore, $(q_{\sup}^*, \bm{d}_{q_{\sup}^*}^*)$ is an optimal
solution of (\ref{opt:original}).
\end{itemize}
\end{theorem}
\noindent
The proof of Theorem \ref{theorem:subproblem} is provided in Appendix
\ref{appendix:subproblem}.

An important consequence of Theorem \ref{theorem:subproblem} is that a
globally optimal solution of (\ref{opt:original}) is obtained by
iteratively solving the subproblem (\ref{opt:subproblem}) with
Theorem \ref{theorem:optimal_solution_Sq}.
Theorem \ref{theorem:subproblem} (b) indicates that we can judge if a
given value of $q > 0$ is smaller, equal to, or greater than the
optimal value $q_{\sup}^*$ of the original problem
(\ref{opt:original}), by comparing the optimal value $x_{q,N}^*$ of
the subproblem (\ref{opt:subproblem}) with the area length $L$.
It also indicates that an optimal placement $\bm{d}^*$ for 
(\ref{opt:original}) is equal to that for the subproblem
(\ref{opt:subproblem}) with $q=q_{\sup}^*$,
which is explicitly calculated from Theorem
\ref{theorem:optimal_solution_Sq} once $q_{\sup}^*$ is obtained.
Theorem \ref{theorem:subproblem} (a) ensures (i) the existence of $q$
such that $x_{q,N}^* = L$ (equivalently $q=q_{\sup}^*$), and (ii) the
monotonicity of $x_{q,N}^*$ with respect to $q$. 
Therefore, a standard bisection method enables us to numerically find
the value of $q_{\sup}^*$, so that we can effectively compute the
optimal placement $\bm{d}^*$ for the original problem
(\ref{opt:original}); Algorithm \ref{alg:algorithm} summarizes
such a procedure.

\begin{algorithm}[t]
{\small
\caption{A global optimization algorithm for (\ref{opt:original}), 
in which the optimal solution $\bm{d}_q^*$ and the optimal value $x_{q,N}^*$ of
(\ref{opt:subproblem}) is computed with Theorem \ref{theorem:optimal_solution_Sq}.}
\label{alg:algorithm}
\begin{algorithmic}[1]
\renewcommand{\algorithmicrequire}{\textbf{Input:}}
\renewcommand{\algorithmicensure}{\textbf{Output:}}

\Require Number of nodes $N$, area length $L$, 
transmission rate function $R(d)$ ($d \geq 0$), and output precision $\epsilon$.

\Ensure An optimal solution $(q_{\sup}^*, \bm{d}^*)$ of (\ref{opt:original}).
\State Find $q_{\mathrm{low}}$ and $q_{\mathrm{up}}$ satisfying
$0 < q_{\mathrm{low}} < q_{\mathrm{up}}$, 
$x_{q_{\mathrm{low}},N}^* \geq L$, and $x_{q_{\mathrm{up}},N}^* < L$.
\While{$q_{\mathrm{up}}-q_{\mathrm{low}} \geq \epsilon$}
\State $q \gets (q_{\mathrm{low}}+q_{\mathrm{up}})/2$.
\State{\textbf{if} $x_{q,N}^* \geq L$ \textbf{then} $q_{\mathrm{low}} \gets q$
\,
\textbf{else} $q_{\mathrm{up}} \gets q$}.
\EndWhile
\State $q_{\sup}^* \gets (q_{\mathrm{low}}+q_{\mathrm{up}})/2$
and $\bm{d}^* \gets \bm{d}_{q_{\sup}}^*$.
\end{algorithmic}
}
\end{algorithm}

Before closing this section, we conduct further investigations
on mathematical structures of the obtained optimal solution.
Let $q_{\sup}^*(L,N)$ ($L > 0$, $N=1,2,\ldots$) denote the optimal
value $q_{\sup}^*$ of the normalized throughput limit, represented as a
function of the area length $L$ and the number of relay nodes $N$ for
a fixed transmission rate function $R$.

\begin{lemma}
\label{lemma:q_sup-decreasing}
For a fixed $N$ ($N = 1,2,\ldots$), $q_{\sup}^*(L,N)$ is a strictly
decreasing function of $L$.
\end{lemma}
\begin{proof}
Lemma \ref{lemma:q_sup-decreasing} immediately follows from
Theorem \ref{theorem:subproblem} (a) and (b): $q_{\sup}^*(N,L)$ for
fixed $N$ is determined by the unique solution of $x_{q,N}^* = L,\, q > 0$,
which is strictly decreasing with $L$.
\end{proof}

\begin{theorem}
\label{theorem:g_q-changepoint-L}
Let $L_0 := g_{q_0}^{-1}(0)$, where $q_0$ is defined in Lemma
\ref{lemma:g_q-changepoint-q}. For any $N = 1,2,\ldots$, we have
\begin{align}
g_{q_{\sup}^*(L,N)}^{-1}(0) \geq g_{q_{\sup}^*(L,N)}(0),
\;\;
\forall L \leq L_0,
\quad
g_{q_{\sup}^*(L,N)}^{-1}(0) < g_{q_{\sup}^*(L,N)}(0),
\;\;
\forall L > L_0,
\end{align}
i.e., if $L \leq L_0$, an optimal node placement for (\ref{opt:original})
is given by the form (\ref{eq:d^*0}), and otherwise it is given by the
form (\ref{eq:d^*}), regardless of the number of nodes $N$.
\end{theorem}
\begin{proof}
We readily obtain Theorem \ref{theorem:g_q-changepoint-L}
from Theorem \ref{theorem:optimal_solution_Sq} (i),
Lemma \ref{lemma:g_q-changepoint-q}, Theorem \ref{theorem:subproblem},
and Lemma \ref{lemma:q_sup-decreasing}, noting that
$q_{\sup}^*(L,N) = q_0
\;\Leftrightarrow\;
L = x_{q_0, N}^* = g_{q_0}^{-1}(0)$.
\end{proof}

Theorem \ref{theorem:g_q-changepoint-L} shows that if the area length
$L$ is smaller than or equal to $L_0$, the normalized throughput limit
$q_{\sup}^*(L,N)$ is not affected by the number of nodes $N$, i.e., 
\textit{no performance gain can be obtained by increasing the number
of relay nodes} in that case. On the other hand, if $L \geq L_0$,
we can verify from Theorem \ref{theorem:subproblem} (a) and 
Theorem \ref{theorem:optimal_solution_Sq} (ii) that $q_{\sup}^*(L,N)$
strictly increases with the number of nodes $N$.

Finally, we provide a further characterization of the sequence $d_1^*, d_2^*,
\ldots, d_N^*$ defined by (\ref{eq:d-recursion}),
restricting our attention to the case $g_q^{-1}(0) > g_q(0)$
(i.e., $L > L_0$ in view of Theorem \ref{theorem:g_q-changepoint-L}). 
As shown in Theorem \ref{theorem:optimal_solution_Sq},
it is optimal to place relay nodes with ascending node intervals in
this case. In other words, if one determines node intervals in the
reversed order (i.e., the interval between $N$th and $(N-1)$st nodes
is determined first), the sequence of optimal node intervals $d_N^*,
d_{N-1}^*, \ldots, d_1^*$ is decreasing. The following theorem shows that in the
optimal placement, the decrease in node intervals is at least exponentially fast:
\begin{theorem}\label{theorem:exponential-decay}

Let $\gamma$ denote a real number
given by
\begin{equation}
\gamma = 1 + (g_q^{-1})'(g(0)) = 1 + \frac{1}{g_q'(0)},
\end{equation}
where $f'$ denotes the derivative of function $f$.
If $g_q^{-1}(0) > g_q(0)$, then $0 < \gamma < 1$ and 
\begin{itemize}
\item[(i)] if $(g_q^{-1})'(0) > -1$, we have
$d_i^* \leq \gamma^{N-i}d_N^*$ ($i = 1,2,\ldots,N-1$),
and
\item[(ii)] if $(g^{-1})'(0) \leq -1$, we have
$
d_i^* \leq \gamma^{N-i-1}d_{N-1}^*,
$ ($i = 1,2,\ldots,N-2$).
\end{itemize}
\end{theorem}

\noindent
The proof of Theorem \ref{theorem:exponential-decay}
is provided in Appendix \ref{appendix:exponential-decay}.

\section{Performance Evaluation and Extension}
\label{sec:numerical}

In this section, we evaluate the performance of the obtained optimal
solution. We first present extensive numerical experiments to
illustrate the effectiveness of the optimal solution, focusing on the
one-dimensional SOWN discussed so far. We then provide an example of
an optimal relay placement problem for a \emph{two-dimensional} SOWN,
which demonstrates how our result can be extended to a more general
situation.

Throughout this section, we employ $R(d)$ given in (\ref{eq:R-Shannon})
and the following SNR equation \cite{Giles05}:
\begin{equation}
\mathrm{SNR}(d)
= 
\frac{P_{\mathrm{t}}D^2 \cos{\varphi}}{4(\tan^2 \theta)P_{\mathrm{n}}}
\cdot
\frac{e^{-Kr}}{(\epsilon+r)^2},
\label{eq:SNR-for-eval}
\end{equation}
where $P_{\mathrm{t}}$ denotes the transmitter power,
$P_{\mathrm{n}}$ denotes the noise power,
$D$ denotes the receiver aperture diameter,
$\varphi$ denotes the angle between the optical axis of the receiver 
and the line-of-sight between the transmitter and the receiver,
$\theta$ denotes the half-angle transmitter beamwidth,
$K$ denotes the beam attenuation coefficient,
and $\epsilon$ denotes the constant introduced in (\ref{eq:SNR}).
For the noise power $P_{\mathrm{n}}$, we employ a constant value
representing thermal noise, assuming that the contribution of shot
noise to $P_{\mathrm{n}}$ is negligible due to the small power
of received optical signals. 

Unless otherwise mentioned, we use parameter values summarized in
Table \ref{table:parameters} as the default values. 
We consider three different values for the beam attenuation
coefficient $K$, reflecting its dependence on the light
wavelength \cite{Smith} (we restrict our attention to the case of pure
water, based on empirical evidence \cite{Riccobene07} in the deep sea).
We also set the area length $L = 500$ [m] unless otherwise mentioned.

\begin{table}
\caption{Default parameter values in numerical experiments.}
\label{table:parameters}
\centering
\begin{tabular}{|c||c|c|}
\hline
Symbol & Unit & Value
\\
\hhline{|=#=|=|}
$P_{\mathrm{t}}$ & Watt & $0.5$
\\
\hline
$P_{\mathrm{n}}$ & Watt & $2\times 10^{-6}$
\\
\hline
$D$ & Meter & $0.2$
\\
\hline
$\varphi$ & Degree & $10$
\\
\hline
$\theta$ & Degree & $10$
\\
\hline
$W$ & Hz & $5 \times 10^8$
\\
\hline
\multirow{3}{*}{$K$} & \multirow{3}{*}{1/Meter} 
&
$\lefteqn{\mbox{Red light (650 [nm]):}}$\hspace{10.0em}$3 \times 10^{-1}$
\\
& & 
$\lefteqn{\mbox{Green light (550 [nm]):}}$\hspace{10.0em}$7\times 10^{-2}$
\\
& & 
$\lefteqn{\mbox{Blue light (450 [nm]):}}$\hspace{10.0em}$2\times 10^{-2}$
\\
\hline
$\epsilon$ & Meter & 1
\\
\hline
\end{tabular}
\end{table}

\subsection{Performance Evaluation of one-dimensional SOWCs}

We start with providing an example of the optimal, non-constant node
intervals we have obtained. Fig.\ \ref{fig:green_placement}
illustrates the optimal node placement for the case of green light 
and Fig.\ \ref{fig:green_distance} shows the corresponding sequence of
optimal node intervals $(d_i^*)_{i=1,2,\ldots,N}$ 
(i.e., Fig.\ \ref{fig:green_distance} plots the spacings of the
placement in Fig.\ \ref{fig:green_placement}). 
We observe that the optimal node interval $d_i^*$ is increasing with
$i$ and that there is a large difference between the values of $d_1^*$ and $d_N^*$.
Fig.\ \ref{fig:bitrate} shows the maximum normalized throughput limit
$q_{\sup}^*$ (achieved by the optimal relay placement) as a
function of the number $N$ of relay nodes for the three
wavelengths. We observe that adding a few relay nodes to the network
drastically expands the stability region of the system. 
Fig.\ \ref{fig:blue_L} shows the maximum normalized throughput
limit $q_{\sup}^*$ as a function of the area length $L$ for the 
blue light. For large values of $L$, we observe that $q_{\sup}^*$
exponentially decreases as $L$ increases. Furthermore, for relatively
small values of $L$, only little difference can be seen between the
values of  $q_{\sup}^*$ for the number of nodes $N=5$, $N=10$, and
$N=20$, and they coincide each other for quite small $L$, i.e., no
performance gain can be obtained by using additional relay nodes (cf.\
Theorem \ref{theorem:g_q-changepoint-L}).

\begin{figure*}[tbp]
\centering
\includegraphics[scale=1.0]{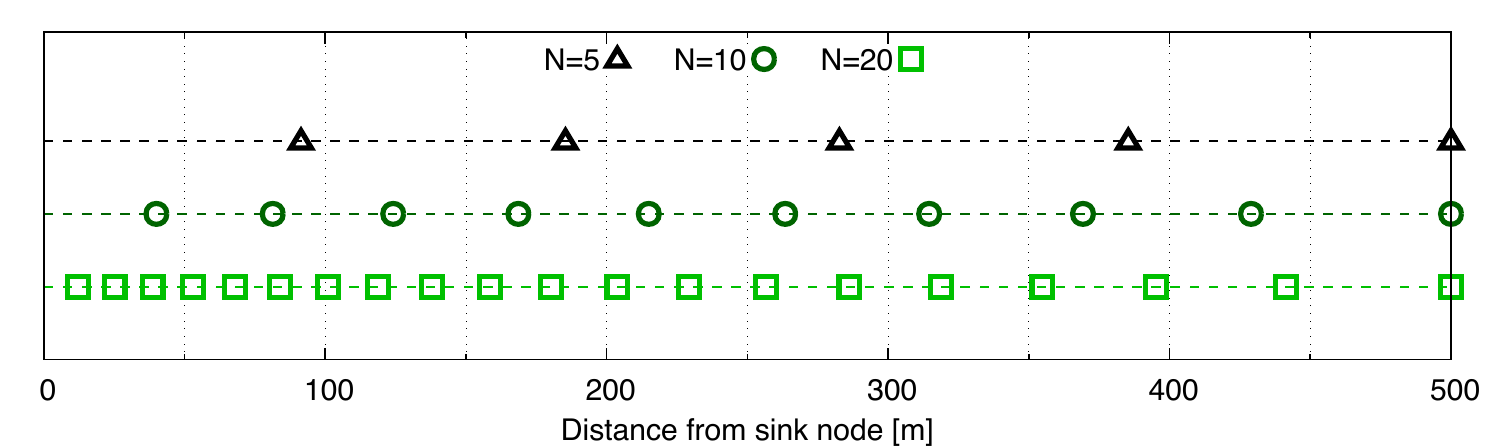}
\vspace{-2ex}
\caption{An illustration of the optimal placement of relay nodes for
the green light, where each symbol represents a relay node.}
\label{fig:green_placement}
\end{figure*}

\begin{figure}[tbp]
\begin{minipage}[t]{0.48\textwidth}
\centering
\includegraphics[scale=1.0]{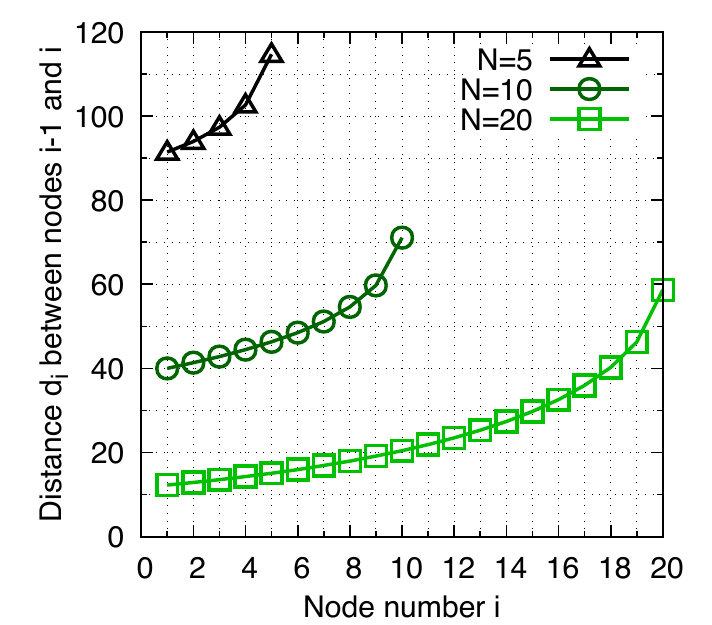}
\vspace{-2ex}
\caption{Optimal distances $\{d_i\}_{i=1,2,\ldots,N}$ between relay
nodes for the green light.}
\label{fig:green_distance}
\end{minipage}
\;\;
\begin{minipage}[t]{0.48\textwidth}
\includegraphics[scale=1.0]{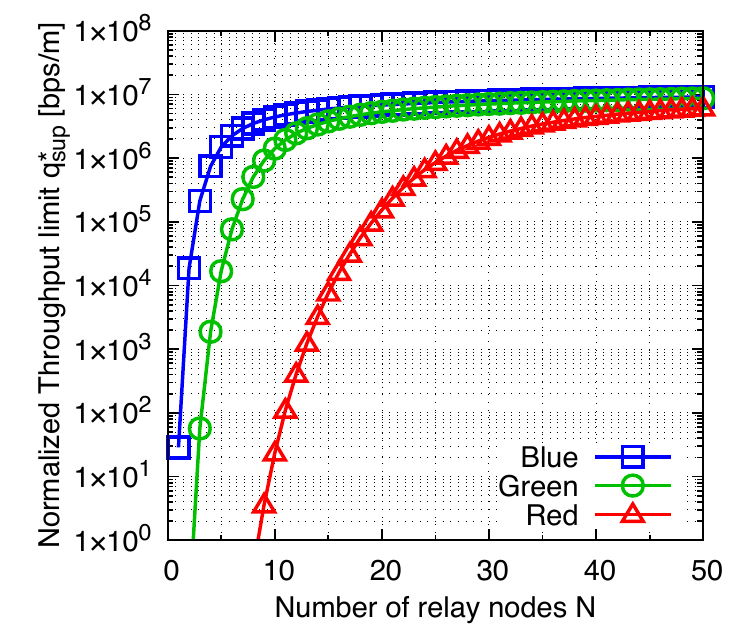}
\vspace{-2ex}
\caption{Impact of the number of relay nodes $N$ on the 
normalized throughput limit $q_{\sup}^*$.}
\label{fig:bitrate}
\end{minipage}
\end{figure}

\begin{figure}[p]
\centering
\begin{minipage}[t]{0.48\textwidth}
\includegraphics[scale=1.0]{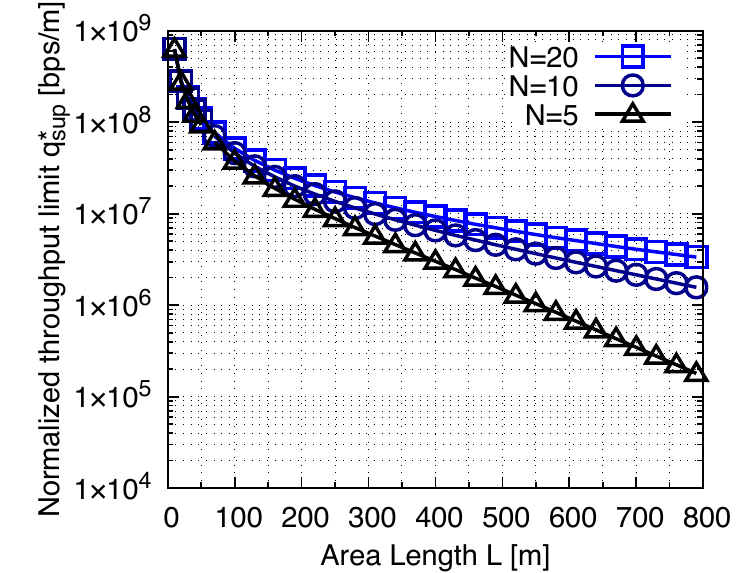}
\vspace{-2ex}
\caption{The normalized throughput limit $q_{\sup}^*$
as a function of coverage length $L$.}
\label{fig:blue_L}
\end{minipage}
\;\;
\begin{minipage}[t]{0.48\textwidth}
\includegraphics[scale=1.0]{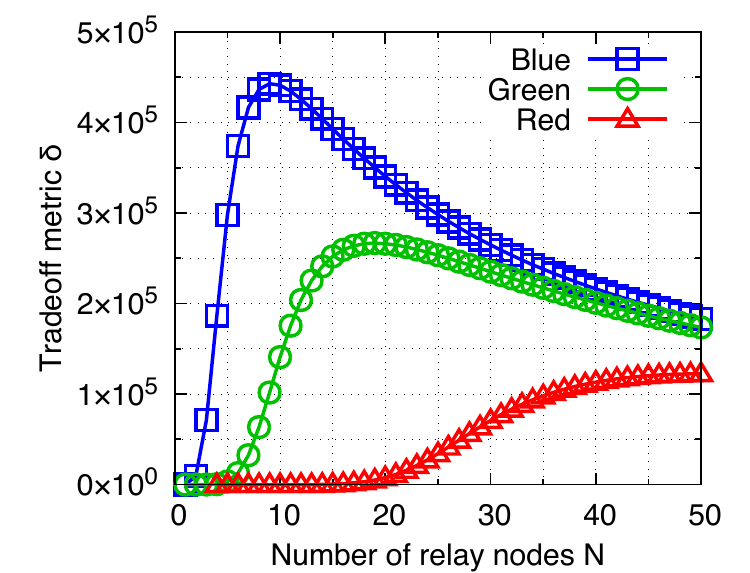}
\vspace{-2ex}
\caption{The tradeoff metric $\delta$ as a function of the
number $N$ of relay nodes.}
\label{fig:tradeoff}
\end{minipage}
\mbox{}
\vspace{0.5ex}
\\
\centering
\begin{minipage}[t]{0.48\textwidth}
\includegraphics[scale=1.0]{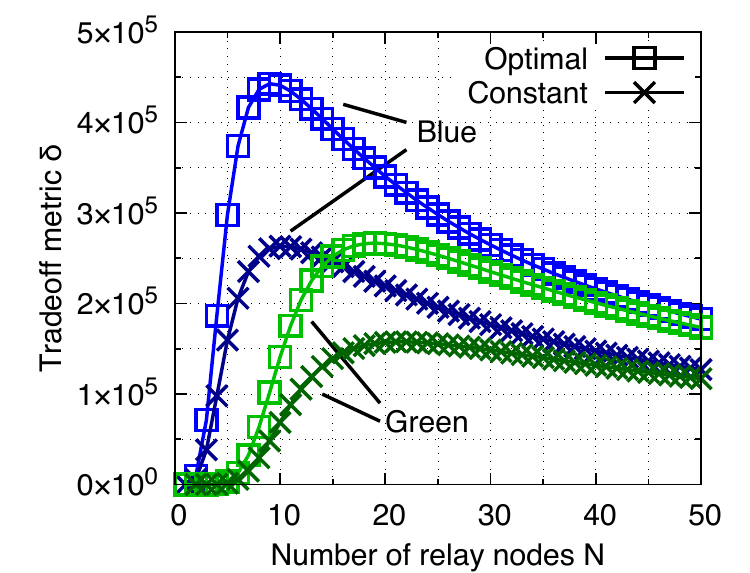}
\vspace{-2ex}
\caption{Performance improvement gained by the optimal solution,
compared to the node placement with constant intervals.}
\label{fig:comparison}
\end{minipage}
\;\;
\begin{minipage}[t]{0.48\textwidth}
\includegraphics[scale=1.0]{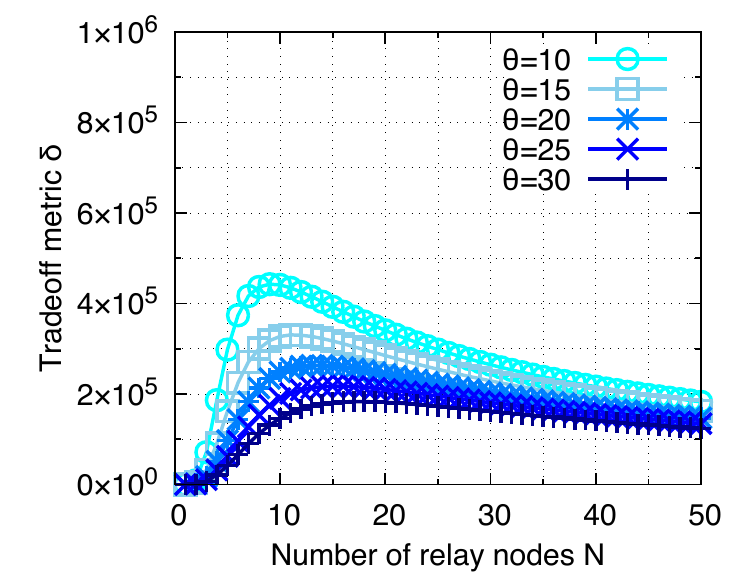}
\vspace{-2ex}
\caption{Effect of the beam width $\theta$ on the tradeoff metric $\delta$.}
\label{fig:divergence}
\end{minipage}
\mbox{}
\vspace{0.5ex}
\\
\begin{minipage}[t]{0.48\textwidth}
\includegraphics[scale=1.0]{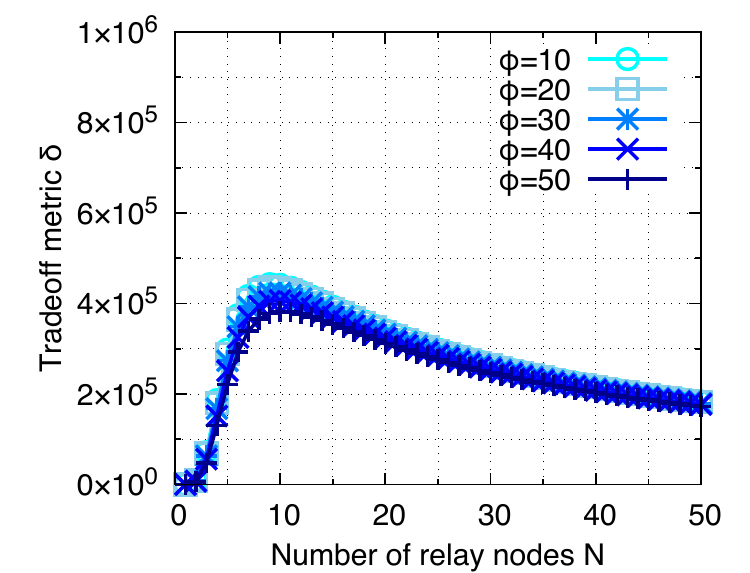}
\vspace{-2ex}
\caption{Effect of the misalignment $\phi$ on the tradeoff metric $\delta$.}
\label{fig:misalignment}
\end{minipage}
\;\;
\begin{minipage}[t]{0.48\textwidth}
\includegraphics[scale=1.0]{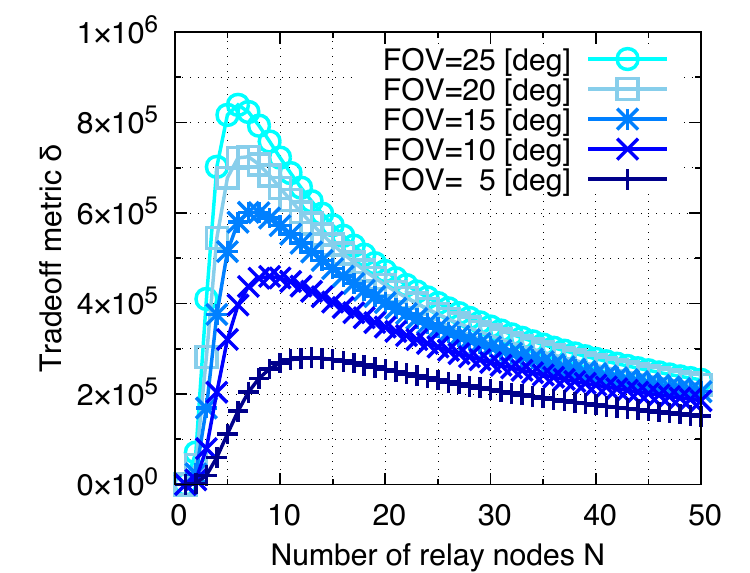}
\vspace{-2ex}
\caption{Effect of the receiver FOV on the tradeoff metric $\delta$.}
\label{fig:FOV}
\end{minipage}
\end{figure}

As shown in Fig.\ \ref{fig:bitrate}, $q_{\sup}^*$ is a concave
function of $N$: the impact of adding a relay node on improvement in
$q_{\sup}^*$ decreases with an increase in the number $N$ of relay nodes. 
The saturation of $q_{\sup}^*$ with an increase in $N$ is due to the
fact that link capacity is inherently bounded above by $R(0)$, so that 
$q_{\sup}^*$ cannot exceed $R(0)/L$ ($\simeq 1.13 \times 10^7$ in the
settings of Fig.\ \ref{fig:bitrate}).
Therefore, a reasonable number of relay nodes can be determined by taking
the cost-performance tradeoff into consideration.
To discuss the tradeoff between the number of nodes and the system
performance, we introduce \emph{a tradeoff metric} $\delta := q_{\sup}/N$.
By definition, $\delta$ represents the (normalized) throughput of the
system \emph{per} relay node. Therefore, the number of relay nodes
$N$ maximizing $\delta$ is optimal in the sense that it
maximizes the cost-performance ratio.
See Fig.\ \ref{fig:tradeoff}, where the tradeoff metric $\delta$ is
plotted as a function of $N$ for the three types of wavelengths.
We observe that the optimal number of relay nodes depends on the
light-wavelength and that the use of blue light is far more effective
than that of green and red lights.

Fig.\ \ref{fig:comparison} compares the optimal placement with the
constant-interval placement $d_i = L/N$ ($i= 1,2,\ldots,N$),
in terms of the tradeoff metric $\delta$; note here that 
between two different placements with the same $N$, the ratio of
$\delta$ is equal to that of the normalized throughput limit
$q_{\sup}$ itself. From Fig.\ \ref{fig:comparison} we observe that 
significant performance improvement is gained by using the optimized,
non-constant node intervals.

We next discuss the effect of several practical aspects of the UOWC
channel on the system performance. For brevity, we present the results
focusing on the case of blue light. Fig.\ \ref{fig:divergence},
Fig.\ \ref{fig:misalignment}, and Fig.\ \ref{fig:FOV} respectively
show the effects of the beam width $\theta$, the misalignment $\phi$,
and receivers' field-of-view (FOV) on the tradeoff metric $\delta$,
where we assume the focal length $F=0.6$ [m] (note that these figures
have different scales from Fig.\ \ref{fig:tradeoff} and Fig.\ \ref{fig:comparison}).
We observe that the beam width and receivers' FOV have a significant
impact on the system performance, while the misalignment $\phi$ has a
less impact on it. This result suggests that (i) improving the receiver
FOV is particularly of great importance in developing optical devices
for SOWNs and (ii) narrowing the beam width (as long as the
line-of-sight (LOS) link is maintained) can effectively increase the
system performance.

In general, underwater nodes may have uncertainty in their positioning
due to the localization error. Fig.\ \ref{fig:uncert} shows the effect
of such uncertainty on the system performance $\delta$, where 
the positions $x_2,x_3,\ldots,x_{N-1}$ of intermediate relay nodes are
perturbed by independent Gaussian noise with mean zero and standard
deviation $\sigma$. We observe that the optimal placement still
attains a good performance even with the localization error. Also, we
see that the optimal number of nodes in terms of the cost-performance
ratio is invariant regardless of $\delta$.

Finally, we compare the proposed SOWN with a conventional UOWN with
vertical relays. Suppose that the seafloor is covered by
$N_{\mathrm{L}}$ relay nodes each of which collects data packets from
an interval of length $L/N_{\mathrm{L}}$ and relays the packets to
a tandem network of $N_{\mathrm{V}}$ vertically aligned relays with
interval $V/N_{\mathrm{V}}$, where $V$ denotes the depth of the
seafloor from the surface of the sea. In this vertical network, 
the normalized throughput limit is given by 
$q_{\sup} = \max\{q > 0;\, (L/N_{\mathrm{L}})q 
\leq R(V/N_{\mathrm{V}})\} = N_{\mathrm{L}}R(V/N_{\mathrm{V}})/L$.
In Fig.\ \ref{fig:vertical}, $q_{\sup}$ of the vertical UOWN is
plotted as a function of the total number of relay nodes
$N = N_{\mathrm{L}}(N_{\mathrm{V}}+1)$ for a case with
$V=3000$ [m] and $N_{\mathrm{L}}=5$, where $q_{\sup}$ in our proposed
SOWN with $N=10$ is also plotted as a reference. As shown in the figure,
to achieve similar performance to the proposed SOWN only with $N=10$,
the vertical UOWN requires at least $N=150$ relay nodes for the blue
light and more than $N=200$ relay nodes for the green light,
which highlights the efficiency of the proposed scheme in collecting
data from deep sea.

\begin{figure}[tbp]
\centering
\begin{minipage}[t]{0.48\textwidth}
\centering
\includegraphics[scale=1.0]{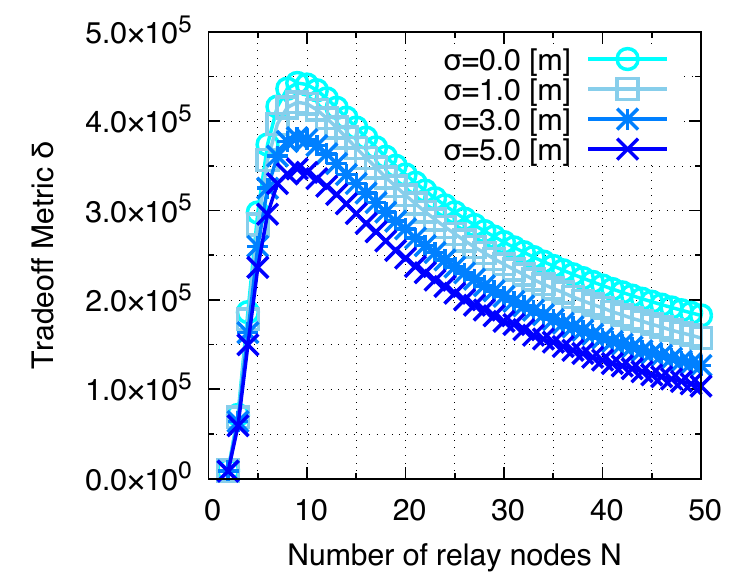}
\caption{Effect of localization uncertainty on the tradeoff metric
$\delta$.}
\label{fig:uncert}
\end{minipage}
\;\;
\begin{minipage}[t]{0.48\textwidth}
\centering
\includegraphics[scale=1.0]{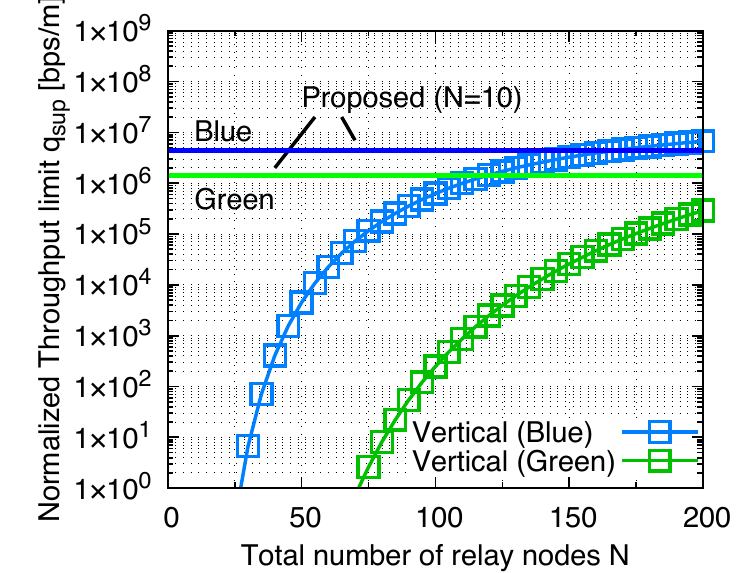}
\caption{Comparison of the proposed SOWN with a conventional vertical
UOWN ($N_L=5$).}
\label{fig:vertical}
\end{minipage}
\mbox{}
\\
\centering
\begin{minipage}[t]{0.48\textwidth}
\includegraphics[scale=0.4]{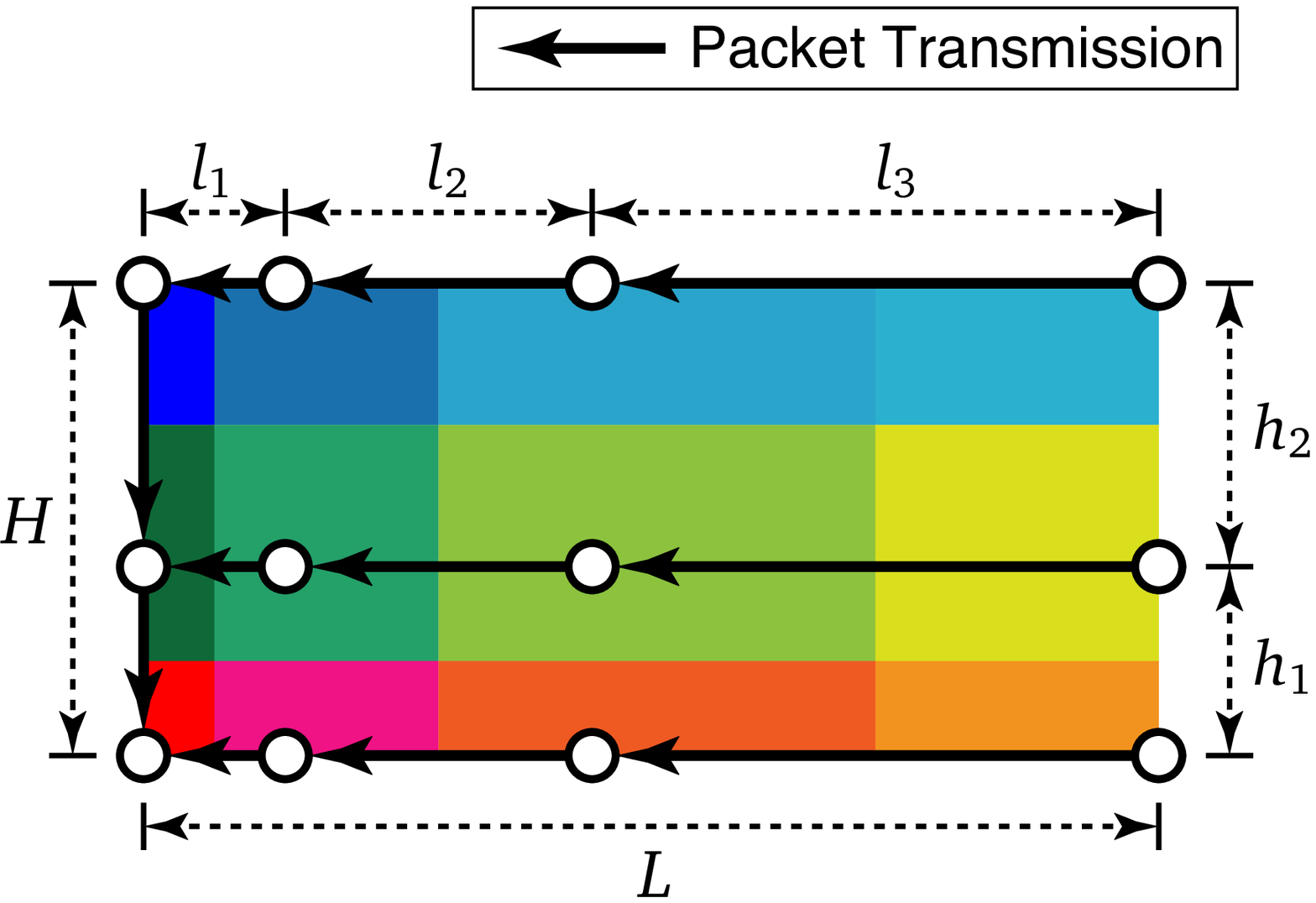}
\caption{A two-dimensional SOWN with $N_{\mathrm{L}} = 4$ and
$N_{\mathrm{H}} = 3$.}
\label{fig:2dim}
\end{minipage}
\;\;
\begin{minipage}[t]{0.48\textwidth}
\includegraphics[scale=1.0]{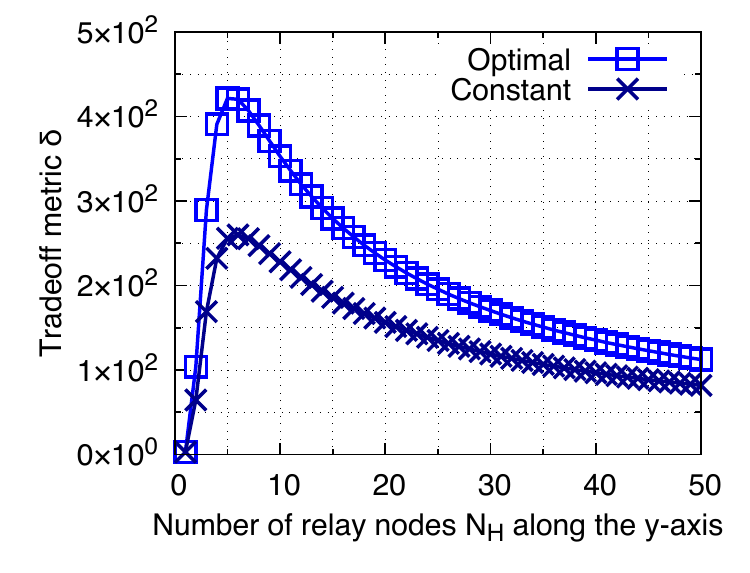}
\caption{Comparison of the optimal and constant relay intervals
in the two-dimensional case ($N_{\mathrm{L}} = 6$).}
\label{fig:2d-numeric}
\end{minipage}
\end{figure}

\subsection{Extension to a two-dimensional SOWN}

The optimization procedure we have developed can be extended to a
two-dimensional SOWN as follows. Suppose that a two-dimensional area
$\mathcal{L} \times \mathcal{H} \subset \mathbb{R}^2$ of a seafloor is
covered by $N$ relay nodes, where $\mathcal{L} = [0,L]$ as before and
$\mathcal{H} := [0,H]$ ($H > 0$). We assume that the sink node is
placed at the origin $(0,0)$ and relay nodes are to be placed in a
grid with non-constant spacings; the $i$th spacing along the $x$-axis
is denoted by $\ell_i$ and the $j$th spacing along the $y$-axis is
denoted by $h_j$. See Fig.\ \ref{fig:2dim} for an illustration.
Let $N_{\mathrm{L}}$ (resp.\ $N_{\mathrm{L}}$) denote the number
of nodes placed along the $x$-axis (resp.\ $y$-axis) for each row
(resp.\ column). Note that the total number of relay nodes in the
network (excluding the sink node) is given by 
$N = N_{\mathrm{L}}N_{\mathrm{H}} - 1$. 
Hereafter we refer to the node placed at the $i$th column ($i =
0,1,\ldots,N_{\mathrm{L}}-1$) and the $j$th row ($j =
0,1,\ldots,N_{\mathrm{H}}-1$) as the ($i,j$)th node, 
where the $(0,0)$th node denotes the sink node. By definition, the
position $(x_i, y_i)$ of the ($i,j$)th node is written as
$(x_i, y_i) = (\sum_{k=1}^i \ell_k, \sum_{k'=1}^j h_{k'})$. 
Similarly to the one-dimensional case, we assume that generation
times of data packets follow a general point process with intensity
$\lambda$, and the generation points of those packets are distributed
uniformly on $\mathcal{L} \times \mathcal{H}$.
The normalized traffic load $q$ is then defined as (cf.\
(\ref{eq:rho_n-def}))
\begin{equation}
q = \frac{\lambda B}{LH},
\end{equation}
where $B$ denotes the mean data size as before.
Each packet is collected by the node nearest from its generation
point and delivered to the sink node with multi-hop
transmission. The cover area $\mathcal{C}_{i,j}$ of the ($i,j$)th
node is then its two-dimensional Voronoi cell, which is a rectangle as
depicted in Fig.\ \ref{fig:2dim}. 
Also, the traffic intensity $\rho_{i,j}(q,\bm{x},\bm{y})$ of external
arrivals to the ($i,j$)th node is given by
$\rho_{i,j}(q,\bm{x},\bm{y}) = q|\mathcal{C}_{i,j}|$;
note here that cover areas $\mathcal{C}_{i,j}$ are, by definition,
determined by the node placement ($\bm{x},\bm{y}$).

To formulate an optimal relay placement problem in the two-dimensional
case, we have to specify the routing paths. Here we concentrate on the
basic routing policy shown in Fig.\ \ref{fig:2dim}: each packet is
first transmitted along the $x$-axis to the left-most node and then 
transmitted along the $y$-axis to the sink node.
In this setting, the one-dimensional optimal relay placement problem
(\ref{opt:original}) maximizing the normalized throughput limit
$q_{\sup}$ readily extends to the two-dimensional case:
\begin{align}
\lefteqn{
\hspace{-0.3em}
\underset{
q \in \mathbb{R},\, 
\bm{\ell} \in \mathbb{R}^{N_{\mathrm{L}}-1},\,
\bm{h} \in \mathbb{R}^{N_{\mathrm{H}}-1},\,
}{\mathrm{maximize}}
\;\;
q
}
\nonumber
\\
&
\hspace{2.0em}
\mathrm{s.t.}\hspace{-0.5em}
&&
R(\ell_i)
-
\max_{j \in \mathcal{N}_{\mathrm{H}}}\left(
\frac{h_j + h_{j+1}}{2}
\right)
\left(
\frac{q\ell_i}{2}
-
\sum_{n=i+1}^N q \ell_n 
\right)
\geq 0,
\quad
i \in \mathcal{N}_{\mathrm{L}},
\nonumber
\\
&&&
R(h_j)
-
L
\left(
\frac{qh_j}{2}
-
\sum_{n=j+1}^N q h_n 
\right)
\geq 0,
\;\;
j \in \mathcal{N}_{\mathrm{L}},
\nonumber
\\
&&& 
q \geq 0, \;\;
\sum_{i=1}^{N_{\mathrm{L}}} \ell_i = L, 
\;\;
\sum_{j=1}^{N_{\mathrm{H}}} h_j = H, 
\nonumber
\;\;
\ell_i \geq 0,
\,
i \in \mathcal{N}_{\mathrm{L}},
\;\;
h_j \geq 0,
\,
j \in \mathcal{N}_{\mathrm{H}},
\tag{$\mbox{U}_{\mathrm{2D}}$}
\label{opt:2dim}
\end{align}
where $\mathcal{N}_{\mathrm{L}} := \{1,2,\ldots,N_{\mathrm{L}}-1\}$, 
$\mathcal{N}_{\mathrm{H}} := \{1,2,\ldots,N_{\mathrm{H}}-1\}$,
and $h_{N_{\mathrm{H}}} := 0$.

Observe that for a fixed spacings $\bm{\ell}$ (resp.\ $\bm{h}$) along the $x$-axis 
(resp.\ $y$-axis), the optimization problem (\ref{opt:2dim}) reduces
to the one-dimensional problem (\ref{opt:original}) by replacing
$R(\cdot)$ with $R(\cdot)/L$ (resp.\ $R(\cdot)/\max_{j \in
N_{\mathrm{H}}}((h_j+h_{j+1})/2)$).
Assuming that communication links along the $y$-axis become
bottlenecks of the system (we shall shortly come back to this point),
we thus obtain optimal spacings $(\bm{\ell}^*, \bm{h}^*)$ by first solving
(\ref{opt:original}) with $R(\cdot)$ replaced by $R(\cdot)/L$ to
obtain $\bm{h}^*$ and then solving (\ref{opt:original}) with
$R(\cdot)$ replaced by $R(\cdot)/\max_{j \in
N_{\mathrm{H}}}((h_j^*+h_{j+1}^*)/2)$ to obtain $\bm{\ell}^*$. 
Write the optimal value obtained at the first step as 
$q_y^* \cdot L$ and that obtained at the second step as
$q_x^* \cdot \max_{j \in N_{\mathrm{H}}} (h_j^* + h_{j+1}^*)/2$.
The achieved objective value is then given by $\min(q_x^*, q_y^*)$, 
as it is the maximum $q$ satisfying all the constraints. That is, we can ensure
that the above-mentioned assumption (communication links along the $y$-axis
are bottlenecks) is indeed satisfied by checking if $q_y^* < q_x^*$ 
is satisfied. Because $q_x^*$ increases with $N_L$, we can find the minimum
$N_L$ such that $q_y^* < q_x^*$ holds, and in that case, the maximum
normalized throughput limit $q_{\sup}^*$ (i.e., the optimal value of
(\ref{opt:2dim})) is given by $q_{\sup}^* = q_y^*$. 
Therefore, the above procedure gives the optimal
(i.e., the minimum) number of nodes $N_L$ along the $x$ axis as well
as the optimal spacings $(\bm{\ell}^*, \bm{h}^*)$ of relay nodes.

Similarly to the one-dimensional case, we define the tradeoff metric
$\delta := q_{\sup}/N$. Fig.\ \ref{fig:2d-numeric} compares 
the optimal relay placement with the constant placement
$\bm{l}=(L/N_{\mathrm{L}}, L/N_{\mathrm{L}}, \ldots,
L/N_{\mathrm{L}})$, $\bm{h}=(H/N_{\mathrm{H}}, H/N_{\mathrm{H}}, \ldots,
H/N_{\mathrm{H}})$ for $N_{\mathrm{L}} = 6$ and the blue light with
the parameter values in Table \ref{table:parameters}. We observe that
the optimal relay placement developed in this paper provides 
significant performance improvement also in the two-dimensional network.

\section{Conclusions}
\label{sec:conclusions}

In this paper, we considered an optimal relay placement problem for
a one-dimensional SOWN. We modeled such a network as a queueing network
with a general input process and we formulated the relay placement problem 
whose objective is to maximize the stability region of the whole
system. We showed that this problem has a non-convex feasible region,
whose global optimization is generally a difficult task.
We then developed an algorithm (presented in Algorithm
\ref{alg:algorithm}) to efficiently compute a globally optimal
solution and investigated the mathematical structure of the obtained 
optimal solution.
Through numerical evaluations, we showed that the obtained optimal
solution provides a significant performance gain, compared to the  
conventional constant-interval relay placement.
We further proposed a method to determine a reasonable number of relay
nodes by introducing the tradeoff metric $\delta$, defined as the 
achieved system performance per relay node.
We also presented extensive numerical experiments, where the
proposed method is compared with a conventional vertical relay, and
discussions on several practical aspects of UOWC channels, such as 
the misalignment, the FOV, and the uncertainty in node placement,
are given.

We finally demonstrated how to extend the developed optimal placement
into a two-dimensional SOWN. While we focused on the case of
regular-grid topology and a simple routing policy depicted in Fig.\
\ref{fig:2dim}, there would be various possible other directions for
extensions. For example, the mathematical result shown in Theorem
\ref{theorem:exponential-decay} suggests that it is efficient to
employ relay intervals $d_N, d_{N-1}, \ldots, d_1$ which decreases at
least exponentially fast. Future works include an application of this
insight to two-dimensional networks with more flexible topologies and
sophisticated routing mechanisms that enables us to
deal with occurrences of node and link failure, which is the important
aspect for the reliability of underwater networks as an infrastructure.

\appendix

\section{Proof of Theorem \ref{theorem:optimal_solution_Sq}}
\label{appendix:optimal_solution_Sq}

For ease of presentation, we introduce a slightly generalized problem. 
Let $g : [0,\infty) \to [-\infty,g(0))$ denote a convex function 
with $g(0) > 0$ which is strictly decreasing and continuously
differentiable. For $k = 1,2,\ldots$, we define 
$f_i^{(k)}(\bm{y})$ ($\bm{y} \in \mathbb{R}^k$, $\bm{y} \geq
\bm{0}$) as
\begin{equation}
f_i^{(k)}(\bm{y}) = g(y_i) - \sum_{n=1}^{i-1}y_n,
\;\;
i = 1,2,\ldots,k.
\label{eq:f_i-def}
\end{equation}
Let $\bm{f}^{(k)}(\bm{y}) 
:= (f_1^{(k)}(\bm{y}), f_2^{(k)}(\bm{y}), \ldots, f_k^{(k)}(\bm{y}))^{\top}$
and $u^{(k)}(\bm{y}) := \sum_{i=1}^k y_i$.
We consider the following optimization problem for $k = 1,2,\ldots$:
\begin{align}
\underset{\bm{y} \in \mathbb{R}^k}{\mathrm{maximize}}
\;
u^{(k)}(\bm{y})
\quad
\mathrm{s.t.}
\;
\bm{f}^{(k)}(\bm{y}) \geq \bm{0}, 
\,
\bm{y} \geq \bm{0}.
\tag{$\mbox{P}^{(k)}$}
\end{align}
Since $f_i^{(k)}$ and $u^{(k)}$ are both convex, ($\mbox{P}^{(k)}$) belongs to RCP. 
We can readily verify that ($\mbox{P}^{(N)}$) reduces to (\ref{opt:subproblem}),
letting $d_n = y_{N-n+1}$ ($n = 1,2,\ldots,N$) and $g(x) = g_q(x)$ ($x \geq 0$).

Let $g^{-1}:(-\infty, g(0)] \to [0,\infty)$ denote the inverse function of $g$. 
Because $g$ is assumed to be convex and strictly decreasing,
$g^{-1}$ is also convex and strictly decreasing. 
Note that 
\begin{equation}
g^{-1}(g(0)) = 0,
\quad
g^{-1}(z) > 0,
\;\; 
-\infty < z < g(0).
\label{eq:ginv-sign}
\end{equation}

Below we provide a proof of the following lemma,
which readily implies Theorem \ref{theorem:optimal_solution_Sq}:
\begin{lemma} \label{lemma:optimal_solution}
\begin{itemize}
\item[(i)] If $g^{-1}(0) \geq g(0)$, then the following $\bm{y}^*$
is an optimal solution of ($\mbox{P}^{(k)}$) 
\begin{equation}
\bm{y}^* = (g^{-1}(0), 0, 0, \ldots, 0)^{\top} \in \mathbb{R}^k,
\label{eq:y^*0}
\end{equation}

\item[(ii)] If $g^{-1}(0) < g(0)$, a recursion
\begin{align}
y_1^* &= g^{-1}(0), 
\quad
y_i^* = g^{-1}\left(\sum_{j=1}^{i-1}y_j^*\right),
\;\;
i = 2,3,\ldots,
\label{eq:y_i^*-def}
\end{align}
well defines a sequence $\{y_i^*\}_{i=1,2,\ldots}^*$ such that
\begin{equation}
0 < y_{i+1}^* < y_i^*,
\quad
i = 1,2,\ldots,
\label{eq:y_i-decrease}
\end{equation}
and for $k = 1,2,\ldots$, the following $\bm{y}^*$ is an optimal
solution of ($\mbox{P}^{(k)}$).
\begin{equation}
\bm{y}^* = (y_1^*, y_2^*, \ldots, y_k^*)^{\top} \in \mathbb{R}^k.
\label{eq:y^*}
\end{equation}
\end{itemize}
\end{lemma}

We start with considering the number of zeros that an optimal solution
can have (see Remark \ref{remark:theorem:optimal_solution_Sq}).
Let $\mathcal{A}^{(k)} \subseteq \mathbb{R}^k$ 
denote the set of feasible solutions of ($\mbox{P}^{(k)}$), and let 
$\mathcal{Y}^{(k)} \subseteq \mathcal{A}^{(k)}$ denote
the set of optimal solutions. For $\bm{y} \in \mathbb{R}^k$, 
we define $\kappa(\bm{y})$ as the number of elements 
of $\bm{y}$ which are equal to zero. We define $\phi^{(k)}$ as the
maximum number of zeros in an optimal solution of ($\mbox{P}^{(k)}$):
\begin{equation}
\phi^{(k)} = \max\{\kappa(\bm{y});\, \bm{y} \in \mathcal{Y}^{(k)}\}.
\label{eq:phik-def}
\end{equation}

\begin{lemma} \label{lemma:optimal-phi}
For $k=1,2,\ldots$, the optimal value of ($\mbox{P}^{(k)}$) is equal to that of
($\mbox{P}^{(k-\phi^{(k)})}$):
\begin{align}
\max\{u^{(k)}(\bm{y});\, \bm{y} \in \mathcal{A}^{(k)}\}
=
\max\{u^{(k-\phi^{(k)})}(\bm{y});\, \bm{y} \in \mathcal{A}^{(k-\phi^{(k)})}\}.
\end{align}
\end{lemma}
\begin{proof}
Since the case of $\phi^{(k)}=0$ is trivial, we assume $\phi^{(k)}>0$.
For any $\bm{y} \in \mathbb{R}^k$, let $y_{+,i}$ denote the $i$th
non-zero element of $\bm{y}$. Let 
$\mathcal{Y}_+^{(k)} 
:= \{
(y_{+,1}, y_{+,2}, \ldots, y_{+,k-\phi^{(k)}})^{\top};\,
\kappa(\bm{y})=\phi^{(k)}, \bm{y} \in \mathcal{Y}^{(k)}
\}.
$
It is readily verified that $\mathcal{Y}_+^{(k)} \subseteq
\mathcal{A}^{(k-\phi^{(k)})}$, and therefore
$
\max\{u^{(k)}(\bm{y});\, \bm{y} \in \mathcal{A}^{(k)}\}
=
\max\{u^{(k-\phi^{(k)})}(\bm{y});\, \bm{y} \in \mathcal{Y}_+^{(k)}\}
\leq
\max\{u^{(k-\phi^{(k)})}(\bm{y});\, \bm{y} \in
\mathcal{A}^{(k-\phi^{(k)})}\}.
$
We then obtain Lemma \ref{lemma:optimal-phi} because 
$
\max\{u^{(k-\phi^{(k)})}(\bm{y});\, \bm{y} \in
\mathcal{A}^{(k-\phi^{(k)})}\}
\leq
\max\{u^{(k)}(\bm{y});\, \bm{y} \in \mathcal{A}^{(k)}\}
$
also follows from that
$(\bar{\bm{y}}, 0, 0, \ldots, 0)^{\top} \in \mathbb{R}^k$ is 
a feasible solution of ($\mbox{P}^{(k)}$)
for any $\bar{\bm{y}} \in \mathcal{Y}^{(k-\phi^{(k)})}$.
\end{proof}

\begin{corollary}\label{corollary:k-phi^k}
$\phi^{(k-\phi^{(k)})}=0$ ($k = 1,2,\ldots$).
\end{corollary}
\begin{proof}
Because the case of $\phi^{(k)}=0$ is trivial, we assume
$\phi^{(k)} > 0$. If $\phi^{(k-\phi^{(k)})} > 0$ holds, 
($\mbox{P}^{(k-\phi^{(k)})}$) has an optimal solution 
$\hat{\bm{y}} \in \mathcal{Y}^{(k-\phi^{(k)})}$ 
such that $\kappa(\hat{\bm{y}}) > 0$. It then follows from Lemma 
\ref{lemma:optimal-phi} that $\hat{\bm{y}}_{\mathrm{e}} :=
(\hat{\bm{y}}, 0, 0, \ldots, 0)^{\top} \in \mathbb{R}^k$
is an optimal solution of ($\mbox{P}^{(k)}$). 
This implies $\kappa(\hat{\bm{y}}_e) = \kappa(\hat{\bm{y}}) + \phi^{(k)} >
\phi^{(k)}$, which contradicts the definition (\ref{eq:phik-def})
of $\phi^{(k)}$.
\end{proof}

We can verify that $\nabla\bm{f}^{(k)}(\bm{y})$ is a
lower-triangular matrix with non-zero (negative) diagonal elements.
We thus have $\det(\nabla\bm{f}^{(k)}(\bm{y})) \neq 0$, so that
$\mathrm{rank}(\nabla\bm{f}^{(k)}(\bm{y})) = k$.
Therefore, if $\bar{\bm{y}} \in \mathbb{R}^k$ satisfies
$\bm{f}^{(k)}(\bar{\bm{y}}) = \bm{0}$ and $\bar{\bm{y}} \geq \bm{0}$,
then it is a basic solution of ($\mbox{P}^{(k)}$).
Furthermore, the following Lemma \ref{lemma:positive-basic}
immediately follows from Remark \ref{remark:activate}:

\begin{lemma} \label{lemma:positive-basic}
If $\bar{\bm{y}} \in \mathbb{R}^k$ is a basic solution of
($\mbox{P}^{(k)}$) satisfying $\bar{\bm{y}} > \bm{0}$, then 
$\bm{f}^{(k)}(\bar{\bm{y}}) = \bm{0}$ holds.
\end{lemma}

\begin{lemma}\label{lemma:y-phi^k}
For fixed $k$ ($k=1,2,\ldots$), the followings hold:
\begin{itemize}
\item[(i)] If there exists no vector $\bar{\bm{y}} \in
\mathbb{R}^k$ satisfying $\bm{\bar{\bm{y}}} > \bm{0}$ 
and $\bm{f}^{(k)}(\bar{\bm{y}}) = \bm{0}$, then $\phi^{(k)} > 0$.
\item[(ii)] If $\phi^{(k)} = 0$, then ($\mbox{P}^{(k)}$) has an
optimal solution $\bar{\bm{y}} \in \mathbb{R}^{(k)}$ satisfying $\bar{\bm{y}} >
\bm{0}$ and $\bm{f}^{(k)}(\bar{\bm{y}}) = \bm{0}$.
\end{itemize}
\end{lemma}

\begin{proof}
We first consider (ii). When $\phi^{(k)}=0$, the elements of each
optimal solution of ($\mbox{P}^{(k)}$) are all positive. 
It then follows from Lemma \ref{lemma:RCP-basic} that 
($\mbox{P}^{(k)}$) has an optimal solution $\bar{\bm{y}} >
\bm{0}$ which is also basic. Therefore, we have 
$\bm{f}(\bar{\bm{y}}) = \bm{0}$ from Lemma \ref{lemma:positive-basic},
which proves (ii). 

We next consider (i). The contraposition of (ii) is that
if there exists no optimal solution $\bar{\bm{y}}$ of
($\mbox{P}^{(k)}$) satisfying $\bm{\bar{\bm{y}}} > \bm{0}$ 
and $\bm{f}^{(k)}(\bar{\bm{y}}) = \bm{0}$, then $\phi^{(k)} > 0$.
We thus have (i) from $\mathcal{Y}^{(k)} \subseteq
\mathbb{R}^k$.
\end{proof}

We can readily verify that if $\{y_i^*\}_{i=1,2,\ldots}$ in
(\ref{eq:y_i^*-def}) is well-defined, $\bm{y} = (y_1^*,
y_2^*,\ldots,y_k^*)^{\top}$ is the unique solution of 
$\bm{f}^{(k)}(\bm{y}) = 0$. Note that $y_1^*$ is always well defined,
while $y_i^*$ ($i = 1,2,\ldots$) is not well-defined if
$\sum_{j=1}^{i-1}y_i^* > g(0)$ because the domain of $g^{-1}$ is
$(-\infty, g(0)]$. We then define $N^* \in \{2,3,\ldots\} \cup \{\infty\}$ as
\begin{equation}
N^* = 
\sup\left\{
i \in \{2,3,\ldots\};\,
\sum_{j=1}^{i-2} y_j^* < g(0)
\right\}.
\end{equation}
By definition $y_i^*$ is well-defined at least for $1 \leq i < N^*$.
In addition, if $N^* < \infty$, then $\sum_{j=1}^{N^*-1} y_j^* \geq g(0)$,
so that $g^{-1}(\sum_{j=1}^{N^*-1} y_j^*)$ is either equal to zero or
not well-defined. We thus have
\begin{equation}
y_i^* > 0,
\quad
1 \leq i < N^*.
\label{eq:y_i^*-positive-N*}
\end{equation}
Furthermore, because $g^{-1}$ is a strictly decreasing function,
\begin{equation}
y_i^* < y_{i-1}^*,
\quad
1 \leq i < N^*.
\label{eq:y_i^*-decrease-N*}
\end{equation}
Let $\{z_i^*\}_{1 \leq i < N^*}$ denote a sequence of non-negative numbers defined as
\begin{equation}
z_0^* = 0,
\quad
z_i^* = \sum_{j=1}^i y_j^*,
\;\;
1 \leq i < N^*,
\label{eq:z_i-def}
\end{equation}
for which we have from (\ref{eq:y_i^*-positive-N*}), 
\begin{equation}
z_i^* < z_{i+1}^*,
\quad
0 \leq i < N^*.
\label{eq:z-increasing}
\end{equation}
Additionally, because (\ref{eq:y_i^*-def}) and (\ref{eq:z_i-def}) imply
\begin{align}
y_i^* &= z_i^* - z_{i-1}^*,
\qquad
1 \leq i < N^*,
\label{eq:y-z-diff}
\\
y_i^* &= g^{-1}(z_{i-1}^*),
\qquad
1 \leq i < N^*,
\label{eq:y_i-by-z_i-1}
\end{align}
$z_i^*$ satisfies the following recursion:
\begin{equation}
z_0^* = 0,
\quad
z_i^* = z_{i-1}^* + g^{-1}(z_{i-1}^*),
\;\;
1 \leq i < N^*.
\label{eq:z_i-recursion-base}
\end{equation}
With defining a function $h: (-\infty, g(0)] \rightarrow \mathbb{R}$ as
\begin{equation}
h(z) = z + g^{-1}(z),
\;\;
-\infty < z \leq g(0),
\label{eq:h-def}
\end{equation}
we rewrite (\ref{eq:z_i-recursion-base}) as
\begin{equation}
z_0^* = 0,
\quad
z_i^* = h(z_{i-1}^*),
\;\;
1 \leq i < N^*.
\label{eq:z_i-recursion}
\end{equation}
It is readily verified from (\ref{eq:ginv-sign})
and the definition of $h$ that
\begin{align}
h(z) &> z,
\;\;
-\infty < z < g(0),
\label{eq:h(z)>z}
\\
h(g(0)) &= g(0).
\label{eq:h-fixed-point}
\end{align}

\begin{lemma}
\label{lemma:N^*}
Either $N^* = 2$ or $N^* = \infty$ holds. Specifically, 
if $g^{-1}(0) \geq g(0)$, then $N^* = 2$, and 
otherwise $N^* = \infty$.
\end{lemma}
\begin{proof}
Because $g^{-1}(0) \geq g(0) \Rightarrow N^* = 2$
immediately follows from the definitions of $y_i^*$ and
$N^*$, we consider the case of $g^{-1}(0) < g(0)$ below. We first show
that
\begin{equation}
g^{-1}(0) < g(0) \ \Rightarrow\ g'(0) < -1.
\label{eq:g'0<minus1}
\end{equation}
Since $g$ is assumed to be a convex function, its
derivative $g'$ is a non-decreasing function. 
It then follows that if $g'(0) \geq -1$, then 
$g'(y) \geq -1\; \mbox{for $y \geq 0$}$, so that
\begin{equation}
y_1^* = \int_0^{g^{-1}(0)} \dd y 
\geq \int_0^{g^{-1}(0)} (-g'(y)) \dd y
= g(0),
\end{equation}
i.e., $g'(0) \geq -1 \ \Rightarrow\ y_1^* \geq g(0)$.
We thus obtain (\ref{eq:g'0<minus1}), taking the contraposition.

Below, we proceed by considering two exclusive cases, 
under the assumption $g^{-1}(0) < g(0)$.

\smallskip

\noindent
\textbf{Case 1. $(g^{-1})'(0) > -1$: }
Because $g^{-1} : (-\infty, g(0)] \rightarrow [0,\infty)$ is a convex
function as noted above, its derivative $(g^{-1})'$ 
is a non-decreasing function, so that we have from $(g^{-1})'(0) > -1$, 
\begin{equation}
(g^{-1})'(z) > -1,
\quad
0 \leq z \leq g(0).
\label{eq:ginv>-1}
\end{equation}
It then follows from (\ref{eq:h-def}) that
$h$ is a strictly increasing function. Furthermore,
we obtain from  $g^{-1}(z) > 0$  ($-\infty \leq z < g(0)$) and
(\ref{eq:h-fixed-point}), 
\begin{equation}
0 < h(z) < g(0),
\;\;
0 \leq z < g(0).
\label{eq:h(z)-in-0-g0}
\end{equation}
We then have $N^* = \infty$ from the following relation obtained by
the induction with (\ref{eq:z_i-recursion}) and (\ref{eq:h(z)-in-0-g0}):
\begin{equation}
0 < z_i^* < g(0),
\quad
i = 1,2,\ldots.
\end{equation}

\noindent
\textbf{Case 2. $(g^{-1})'(0) \leq -1$: }
From (\ref{eq:g'0<minus1}) and $(g^{-1})'(g(0)) = 1/g'(0)$,
we have (cf.\ (\ref{eq:ginv>-1}))
\begin{equation}
(g^{-1})'(g(0)) > -1.
\end{equation}
Recall that $g$ is assumed to be continuously differentiable, so that
$(g^{-1})'(y)$ is a continuous function. We can then verify that
there exists $\beta \in [0, g(0))$ satisfying $(g^{-1})'(\beta) = -1$,
using $(g^{-1})'(0) \leq -1$ and the intermediate value theorem.
Note that such $\beta$ is not necessarily unique, because $(g^{-1})'$ is not
necessarily strictly increasing. Instead, the set of such $\beta$
is bounded above, so that its maximum value $\beta^*$ is uniquely
obtained:
\begin{equation}
\beta^* = \max\{\beta \in [0, g(0));\, (g^{-1})'(\beta) = -1\}.
\label{eq:beta-def}
\end{equation}
It is then readily verified that
\begin{align}
(g^{-1})'(z) &\leq -1, 
\qquad 0 \leq z \leq \beta^*,
\label{eq:z<beta}
\\
(g^{-1})'(z) &> -1,
\qquad
\beta^* < z \leq g(0).
\label{eq:z>beta}
\end{align}

We have from (\ref{eq:z<beta}), 
\begin{align}
\beta^* 
= 
\int_0^{\beta^*} \dd z 
&\leq 
\int_0^{\beta^*} (-1) \cdot (g^{-1})'(z) \dd z
=
g^{-1}(0) - g^{-1}(\beta^*)
< g^{-1}(0),
\end{align}
so that (\ref{eq:z_i-recursion-base}) and $g^{-1}(0) < g(0)$ imply
\begin{equation}
\beta^* < z_1^* = g^{-1}(0) < g(0).
\label{eq:z_1-in-beta-g0}
\end{equation}
It follows from (\ref{eq:h-def}) and (\ref{eq:z>beta}) that $h(z)$ is
strictly increasing for $\beta^* < z \leq g(0)$.
Furthermore, (\ref{eq:h(z)>z}) implies $h(\beta^*) > \beta^*$.
We can then verify that (cf.\ (\ref{eq:h(z)-in-0-g0})):
\begin{equation}
\beta^* < h(z) < g(0),
\qquad
\beta^* \leq z < g(0).
\label{eq:h(z)-in-beta-g0}
\end{equation}
We then obtain $N^* = \infty$ because the induction using
(\ref{eq:z_1-in-beta-g0}) and (\ref{eq:h(z)-in-beta-g0}) yields
\begin{equation}
\beta^* < z_i^* < g(0),
\quad
i = 1,2,\ldots.
\label{eq:z_i-in-beta-g0}
\qedhere
\end{equation}
\end{proof}

\begin{lemma}\label{lemma:y^*-solution}
For $k = 2,3,\ldots$, the followings hold:
\begin{itemize}
\item[(i)] If $g^{-1}(0) \geq g(0)$, there exists no
vector $\bar{\bm{y}} \in \mathbb{R}^k$
such that  $\bm{f}^{(k)}(\bar{\bm{y}})=\bm{0}$ and $\bar{\bm{y}} > \bm{0}$.
\item[(ii)] If $g^{-1}(0) < g(0)$, $\bar{\bm{y}} = (y_1^*, y_2^*,\ldots,y_k^*)^{\top}$ is the unique solution of $\bm{f}^{(k)}(\bar{\bm{y}})=\bm{0}$ and $\bar{\bm{y}} > \bm{0}$.
\end{itemize}
\end{lemma}
\begin{proof}
Lemma \ref{lemma:y^*-solution} immediately follows from
(\ref{eq:f_i-def}), (\ref{eq:y_i^*-def}), 
Lemma \ref{lemma:N^*}, and the definition of $N^*$.
\end{proof}

We are now in a position to prove Lemma \ref{lemma:optimal_solution}.

\begin{proof}[Proof of Lemma \ref{lemma:optimal_solution}]

We first consider the case of $g^{-1}(0) \geq g(0)$.
In this case, we have $\phi^{(k)} > 0$ for $k = 2,3,\ldots$
from Lemma \ref{lemma:y-phi^k} (i) and Lemma \ref{lemma:y^*-solution} (i). 
Note that $y_1 = g^{-1}(0)$ is the optimal solution of
($\mbox{P}^{(1)}$), and its optimal value is also equal to
$g^{-1}(0)$. Obviously, we have $\phi^{(2)} = 1$, so that the optimal
value of ($\mbox{P}^{(2)}$) equals to $g^{-1}(0)$.
Owing to Lemma \ref{lemma:optimal-phi}, the optimal value of 
($\mbox{P}^{(3)}$) is then equal to $g^{-1}(0)$, which implies
$\phi^{(3)} = 2$. Therefore, proceeding in the same way,
we can readily show that $\phi^{(k)} = k-1$ and the optimal value of 
$(\mbox{P}^{(k)})$ is equal to $g^{-1}(0)$ for $k = 2,3,\ldots$.
Because (\ref{eq:y^*0}) achieves the optimal value $g^{-1}(0)$ of
$(\mbox{P}^{(k)})$, we obtain Lemma \ref{lemma:optimal_solution}
(i).

We next consider the case of $g^{-1}(0) < g(0)$.
Note first that the well-definedness of $\{y_i\}_{i=1,2,\ldots}^*$ 
and (\ref{eq:y_i-decrease}) have been proved in 
(\ref{eq:y_i^*-positive-N*}), (\ref{eq:y_i^*-decrease-N*}), and
Lemma \ref{lemma:N^*}. It then follows from Lemma \ref{lemma:y^*-solution} 
(ii) that  $(y_1^*, y_2^*,\ldots,y_{k-\phi^{(k)}}^*)^{\top} \in \mathbb{R}^k$ is
the unique solution of  $\bm{f}^{(k-\phi^{(k)})}(\bar{\bm{y}}) = \bm{0}$ 
and  $\bar{\bm{y}} > \bm{0}$. Therefore, from Corollary
\ref{corollary:k-phi^k} and Lemma \ref{lemma:y-phi^k} (ii), 
we can verify that $(y_1^*, y_2^*,\ldots,y_{k-\phi^{(k)}}^*)^{\top}$
is an optimal solution of $(\mbox{P}^{(k-\phi^{(k)})})$ and that
from Lemma \ref{lemma:optimal-phi}, 
$
\max\{u^{(k)}(\bm{y});\, \bm{y} \in \mathcal{A}^{(k)}\}
= \sum_{i=1}^{k-\phi^{(k)}} y_i^*.
$
Because $(y_1^*, y_2^*,\ldots,y_k^*)^{\top} \in
\mathcal{A}^{(k)}$, this equation implies
$\sum_{i=1}^k y_i^* \leq \sum_{i=1}^{k-\phi^{(k)}} y_i^*$.
Therefore, from (\ref{eq:y_i^*-positive-N*}) we have
$\phi^{(k)}=0$, so that $(y_1^*,
y_2^*,\ldots,y_k^*)^{\top}$ is an optimal solution of $(\mbox{P}^{(k)})$.
\end{proof}

\section{Proof of Theorem \ref{theorem:subproblem}}
\label{appendix:subproblem}

We first consider (a). Note first that Theorem
\ref{theorem:optimal_solution_Sq} and Lemma
\ref{lemma:g_q-changepoint-q} imply
\begin{equation}
x_{q,N}^* \geq g_q^{-1}(0), \quad 0 < q < q_0,
\qquad
x_{q,N}^* = g_q^{-1}(0), \quad q \geq q_0,
\label{eq:x_q-q>q_0}
\end{equation}
so that we obtain 
$
\lim_{q \to 0+} x_{q,N}^* 
\geq 
\lim_{q \to \infty} g_q^{-1}(0) = \infty,
$ and
$
\lim_{q \to \infty} x_{q,N}^* 
= \lim_{q \to \infty} g_q^{-1}(0) = 0.
$

For $q \geq q_0$, we have from (\ref{eq:x_q-q>q_0}) that
$x_{q,N}^*$ is continuous and strictly decreasing in $q$
(cf.\ (\ref{eq:g_q-def})). We then assume $0 < q < q_0$. 
Let $h_q: (-\infty, g_q(0)] \to \mathbb{R}$ be defined as
(cf.\ (\ref{eq:h-def}))
\begin{equation}
h_q(s) = s + g_q^{-1}(s),
\quad
s > 0.
\end{equation}
We define $s_{q,i}$ ($i = 1,2,\ldots,N$) as
(cf.\ (\ref{eq:z_i-def}) and (\ref{eq:z_i-recursion}))
\begin{equation}
s_{q,1} = g_q^{-1}(0),
\quad
s_{q,i} = h_q(s_{i-1}),
\;\;
i = 2,3,\ldots,N.
\end{equation}
It is then readily verified from Theorem \ref{theorem:optimal_solution_Sq}
that 
\begin{equation}
x_{q,N}^* = s_{q,N},
\quad
0 < q < q_0.
\label{eq:x_q=s_q}
\end{equation}
$x_{q,N}^*$ is thus a continuous function of $q$. 
By definition, $s_{q,1}$ and $h_q(s)$ (for a fixed $s$)
are strictly decreasing with respect to $q$ (cf.\ (\ref{eq:g_q-def})). 
Furthermore, as shown in the proof of Lemma \ref{lemma:N^*}, 
for a fixed $q$, $h_q(s)$ is strictly increasing with respect to $s$
for $s_{q,1}=g_q^{-1}(0) < s \leq g_q(0)$.
We can then show by induction that
$
s_{q,i} > s_{q',i},
$
($i = 1,2,\ldots,N$) for any $0 < q < q' < q_0$, which and
(\ref{eq:x_q=s_q}) prove that $x_{q,N}^*$ is continuous and strictly
decreasing for $0 < q < q_0$.

What remains is to prove that $x_{q,N}^*$ is continuous at $q=q_0$.
By definition of $q_0$, we have
$\lim_{q \to q_0-}g_q^{-1}(0) = g_{q_0}(0)$,
so that 
$
\lim_{q \to q_0-}s_{q,i} = g_{q_0}(0)
$ ($i= 1,2,\ldots,N$).
$x_{q,N}^*$ is thus continuous at $q=q_0$ because
$
\lim_{q \to q_0+}x_{q,N}^* = x_{q_0}^* = g_{q_0}^{-1}(0) = g_{q_0}(0).
$

We then consider (b). We first show the following relations:
\begin{align}
x_{q,N}^* > L \ \Rightarrow\ q < q_{\sup}^*,
\qquad
x_{q,N}^* &< L \ \Rightarrow\ q > q_{\sup}^*.
\label{eq:if-x_qN>L-or-x_qN<L}
\end{align}
Suppose $x_{q,N}^* > L$ and define 
$\hat{q} := q x_{q,N}^*/L$ and $\hat{d}_i := d_{q,i}^* L/x_{q,N}^*$
($i = 1,2,\ldots,N$), where we have $\hat{q} > q$ and $\hat{d}_i \leq
d_{q,i}^*$. It is then verified that 
$(\hat{q}, \hat{d}_1, \hat{d}_2,\ldots,\hat{d}_N)^{\top}$
is a feasible solution of (\ref{opt:original}), as
$\sum_{i=1}^N \hat{d}_i = L$ and $R(\hat{d}_i) \geq R(d_{q,i}^*)$.
We thus obtain $q_{\sup}^* \geq \hat{q} > q$ from the optimality of
$q_{\sup}^*$, which proves the first relation in (\ref{eq:if-x_qN>L-or-x_qN<L}).
On the other hand, the second relation follows from that
\begin{equation}
x_{q,N}^* < L \ \Rightarrow\ 
\mbox{
For any $\bm{d} \in \mathbb{R}^N$,
$(q, \bm{d})$ is an 
infeasible solution of (\ref{opt:original}),
}
\label{eq:x_qN-infeasible}
\end{equation}
which is proved by contradiction: 
if there exists $\bm{d} \in \mathbb{R}^N$ such that $(q, \bm{d})$ is a
feasible solution of (\ref{opt:original}), $\bm{d}$ is also a feasible
solution of (\ref{opt:subproblem}) satisfying $\sum_{i=1}^N d_i = L$,
contradicting $x_{q,N}^* < L$.

Taking the contrapositions of (\ref{eq:if-x_qN>L-or-x_qN<L}),
we have $q \geq q_{\sup}^* \Rightarrow x_{q,N}^* \leq L$
and $q \leq q_{\sup}^* \Rightarrow x_{q,N}^* \geq L$,
which implies $q = q_{\sup}^* \Rightarrow x_{q,N}^* = L$.
Owing to Theorem \ref{theorem:subproblem} (a), this also implies
$x_{q,N}^* = L \Rightarrow q = q_{\sup}^*$, so that we obtain
the last equivalence relation in (\ref{eq:subproblem}). 
Furthermore, $(q_{\sup}^*, \bm{d}_{q_{\sup}^*}^*)$ is an optimal
solution of (\ref{opt:original}) because it is a feasible solution
with the optimal objective value $q_{\sup}^*$.

Finally, the first and second equivalence relations in (\ref{eq:subproblem})
are immediately obtained, noting that Theorem \ref{theorem:subproblem}
(a) and the last relation in (\ref{eq:subproblem}) imply
$q < q_{\sup}^* \Leftrightarrow x_{q,N}^* > x_{q_{\sup}^*,N}^* = L$
and $q > q_{\sup}^* \Leftrightarrow x_{q,N}^* < x_{q_{\sup}^*,N}^* = L$.
\qed

\section{Proof of Theorem \ref{theorem:exponential-decay}}
\label{appendix:exponential-decay}

We consider the slightly generalized problem ($\mbox{P}^{(k)}$)
considered in Appendix \ref{appendix:optimal_solution_Sq},
assuming $g^{-1}(0) < g(0)$.
It is sufficient to show that under this assumption, 
\begin{equation}
\gamma := 1 + \frac{1}{g'(0)} \in (0,1),
\label{eq:gamma-01-general}
\end{equation}
and that $y_i^*$ defined in (\ref{eq:y_i^*-def}) satisfies
the followings: if $(g^{-1})'(0) > -1$, then
\begin{equation}
0 < y_i^* \leq \gamma^{i-1}g^{-1}(0),
\;\;
i = 1,2,\ldots,
\label{eq:y_i-bound-1}
\end{equation}
and otherwise
\begin{equation}
0 < y_i^* \leq \gamma^{i-2}g^{-1}(g^{-1}(0)),
\;\;
i = 2,3,\ldots.
\label{eq:y_i-bound-2}
\end{equation}
Because (\ref{eq:gamma-01-general}) immediately follows from
(\ref{eq:g'0<minus1}), we show (\ref{eq:y_i-bound-1}) and 
(\ref{eq:y_i-bound-2}) below.

We first consider the case $(g^{-1})'(0) > -1$. As shown in the 
proof of Lemma \ref{lemma:N^*}, $h(z)$ ($0 \leq z \leq g(0)$) is
a strictly increasing function in this case. In addition,
we have $(g^{-1})'(z) \leq (g^{-1})'(g(0))$ ($0 \leq z \leq g(0)$)
because $g^{-1}(z)$ is a convex function.
We then have for any $0 \leq t_1 \leq t_2 \leq g(0)$,
\begin{align}
|h(t_2) - h(t_1)| 
&= t_2 - t_1 + g^{-1}(t_2) - g^{-1}(t_1)
\nonumber
\\
&=
t_2 - t_1  
+ \int_{t_1}^{t_2} (g^{-1})'(t) \dd t
\leq
t_2 - t_1  
+(g^{-1})'(g(0)) \int_{t_1}^{t_2} \dd t
=
\gamma |t_2 - t_1|.
\label{eq:h-contraction}
\end{align}
Furthermore, it follows from (\ref{eq:h-fixed-point}) and (\ref{eq:h(z)-in-0-g0}) 
that $0 \leq h(z) \leq g(0)$ ($0 \leq h(z) \leq g(0)$). 
Therefore, we can verify from (\ref{eq:z_i-recursion}) and Banach
fixed point theorem that $\{z_i^*\}_{i=0,1,\ldots}$ converges to the
unique fixed point of $h(z)$ ($0 \leq z \leq g(0)$) given by
$z = g(0)$. In addition, we have from (\ref{eq:z-increasing}) and
(\ref{eq:h-contraction}),
\begin{align}
y_{i+1}^* &= z_{i+1}^* - z_i^* 
= h(z_i^*) - h(z_{i-1}^*) 
\leq \gamma(z_i^* - z_{i-1}^*)
= 
\gamma y_i^*,
\qquad
i = 1,2,\ldots
\label{y-gamma1}
\end{align}
so that (\ref{eq:y_i-bound-1}) is obtained by induction using
$y_1^* = g^{-1}(0)$. 

We next consider the case $(g^{-1})'(0) \leq -1$.
As shown in the proof of Lemma \ref{lemma:N^*}, $h(z)$ is strictly
increasing for $\beta^* \leq z \leq g(0)$, where $\beta^*$ is defined
in (\ref{eq:beta-def}). We can then show in the same way as
(\ref{eq:h-contraction}) that for any $\beta^* \leq t_1 \leq t_2 \leq
g(0)$, 
\begin{align}
|h(t_2) - h(t_1)| \leq \gamma |t_2 - t_1|.
\label{eq:h-contraction2}
\end{align}
Also, we have $\beta^* \leq h(z) \leq g(0)$ ($\beta^* \leq z \leq g(0)$) 
from (\ref{eq:h-fixed-point}) and (\ref{eq:h(z)-in-beta-g0}). 
It thus follows from (\ref{eq:z_1-in-beta-g0}) and Banach fixed point
theorem that $\{z_i^*\}_{i=0,1,\ldots}$ converges to the unique fixed
point $z = g(0)$ of $h(z)$ in $\beta^* \leq z \leq g(0)$.
Furthermore, in the same way as (\ref{y-gamma1}), 
we have $y_{i+1}^* \leq \gamma y_i^*$ ($i = 2,3,\ldots$)
from (\ref{eq:z-increasing}) and (\ref{eq:h-contraction2}),
so that we obtain (\ref{eq:y_i-bound-2}) noting that $y_2^* =
g^{-1}(g^{-1}(0))$ (see (\ref{eq:y_i^*-def})).
\qed

\end{document}